\documentclass[12pt]{article}
\usepackage{amsmath, amsthm}
\usepackage{amscd}
\usepackage{amssymb}
\usepackage{array}
\usepackage{diagrams}
\usepackage{color}
\usepackage{titling}

\newtheorem{thm}{Theorem}[section]
\newtheorem{theorem}[thm]{Theorem}

\theoremstyle{definition}
\newtheorem{definition}[thm]{Definition}
\newtheorem{example}[thm]{Example}

\newtheorem{remark}[thm]{Remark}

\newcommand{\id}{\relax{\rm 1\kern-.28em 1}}

\newcommand{\R}{\mathbb{R}}
\newcommand{\C}{\mathbb{C}}

\newcommand{\beq}{\begin{equation}}
\newcommand{\eeq}{\end{equation}}

\newcommand{\Hom}{\mathrm{Hom}}

\newcommand{\ri}{\mathrm{i}}
\newcommand{\rid}{\mathrm{id}}
\newcommand{\rss}{\mathrm{(sspaces)}}
\newcommand{\rsm}{\mathrm{(smanifolds)}}
\newcommand{\sets}{\mathrm{(sets)}}

\newcommand{\rspec}{\mathrm{Spec}}
\newcommand{\uspec}{\underline{\mathrm{Spec}}}

\newcommand{\al}{\alpha}
\newcommand{\be}{\beta}

\newcommand{\ga}{\gamma}
\newcommand{\ep}{\epsilon}

\newcommand{\dal}{{\dot\alpha}}
\newcommand{\dbe}{{\dot\beta}}

\newcommand{\ra}{\rightarrow}

\newcommand{\rd}{\mathrm{d}}
\newcommand{\res}{\mathrm{res}}
\newcommand{\red}{\mathrm{red}}
\newcommand{\rspan}{\mathrm{span}}

\newcommand{\cL}{\mathcal{L}}
\newcommand{\cF}{\mathcal{F}}
\newcommand{\cG}{\mathcal{G}}
\newcommand{\cO}{\mathcal{O}}
\newcommand{\cA}{\mathcal{A}}
\newcommand{\cB}{\mathcal{B}}
\newcommand{\cJ}{\mathcal{J}}
\newcommand{\cC}{\mathcal{C}}
\newcommand{\cD}{\mathcal{D}}
\newcommand{\cH}{\mathcal{H}}

\newcommand{\cHom}{\mathcal{H}om}

\newcommand{\csa}{\mathrm{(c \,salgebras)}}
\newcommand{\svs}{\mathrm{(svector \,spaces)}}

\newcommand{\fp}{\mathfrak p}

\begin{document}
\begin{titlingpage}

\title{\hfill {\small IFIC 17-02}\\\hfill\\\hfill\\{\bf Superfields, nilpotent superfields \\and superschemes}}
\author{Mar\'{\i}a A. Lled\'{o}\\\\\texttt{maria.lledo@ific.uv}\\\\
 Departament de F\'{\i}sica Te\`{o}rica,
Universitat de Val\`{e}ncia \\
 and IFIC (CSIC-UVEG)\\
 C/ Dr. Moliner 50, E-46100 Burjassot (Val\`{e}ncia), Spain.}

\date{}

    \maketitle
    \begin{abstract}
       We interpret superfields in a  functorial formalism that explains the properties that  are assumed for them in the physical applications. The starting point of this research was the need to understand in a sound mathematical framework some algebraic constraints imposed on them, but it lead us to revise the very definition of superfield. The constraints that we investigate in the present work give rise to superschemes that, generically,  are not regular, that is, they do not define a standard supermanifold.
    \end{abstract}
\end{titlingpage}

\section{Introduction}\label{introduction-sec}

 Two equivalent approaches are used by mathematicians to understand supervarieties and supermanifolds. The classical one is by means  of sheaves of superalgebras that properly generalize similar definitions of standard algebraic geometry. This approach started with Berezin, Leites, Kostant  and Manin \cite{bl,ko,le,ma}. We will be more precise in Sections \ref{sclarsf-sec}  and \ref{superschemes-sec}, but let us make a rough approximation here.  From the point of view of algebraic geometry, spaces are  described in terms of rings. Essentially, one constructs first a topological space associated to the ring. The category of such spaces is then (contravariantly)  equivalent to the category of rings considered. One can proceed by replacing commutative rings with commutative (or {\it supercommutative}) superrings and most of the geometric structures generalize, although new phenomena appear. The generalization is successful because, even when superalgebras are non commutative algebras, their non commutativity is of a very particular kind, and one does not need indeed to go into the troubles of the full non commutative geometry. Useful reviews of this approach are for example Refs. \cite{dm,va,ccf,fl}.

On the other hand, the functorial  approach to supergeometry was started by Schwarz \cite{as}. A superspace is described as the association of a set to any arbitrary Grassmann algebra $\Lambda$. The elements in each set are called $\Lambda$-points. The sets of  $\Lambda$-points, though, have to behave properly under homomorphisms of superalgebras $\Lambda\rightarrow \Lambda'$. This means that the assignation has to be functorial, so a superspace is a covariant functor from the category of Grassmann algebras to the category of sets.

The two approaches are related. A space $M$ can also be described in terms of the morphisms $S\rightarrow M$ from any other space $S$ in the same category to $M$. Each  morphism  is called an $S$-point of $M$. The association to each $S$ of the set $M(S):=\Hom(S, M)$ is a contravariant functor called the {\it functor of points of $M$}. A morphism between two spaces $f:M\rightarrow N$ can be given as a set of maps   $M(S)\rightarrow N(S)$  between S-points that is functorial in $S$. This is in fact the content of Yoneda's lemma.

 Using the equivalence of categories between spaces and rings, the functor can be given over the rings corresponding to the global sections of the superspaces, so we will have a covariant functor $\cO(S)\rightarrow \Hom (\cO(M),\cO(S))$.

 The construction, with its own peculiarities,  can be carried over to the category of superrings. Then, the functor of points\footnote{A didactic exposition of the functorial approach in modern terms is in the talk by P. Deligne linked in Ref.  \cite{de}.}, when restricted to the subcategory of Grassmann algebras, can be interpreted as the functor of Schwarz \cite{as}. The original idea was developed further in Ref. \cite{av}  where mappings among superspaces, that is, natural transformations of the functors, are connected with the more standard approach of sheaf theory.

Yet another approach to formalize the theory of supermanifolds was proposed some time ago by  Rogers and De Witt \cite{dw,ar}.   The supergeometry was introduced here by generalizing the numbers (real or complex numbers) to {\it supernumbers}, which are elements of a Grassmann algebra, and devising a calculus over them. Many geometric concepts that were being used by physicists in an heuristic definition appeared to have a more sound meaning; indeed, the language used was close to the language used by physicists.

In Ref. \cite{cs} a comparison of the three approaches is done. An important result is that the Rogers--De Witt point of view can be reconciled with the mathematical one if one assumes that, rather than working with a particular, infinite dimensional Grassmann algebra, the geometric structures are defined functorially over the whole category of superalgebras. It seems to us that the functorial behavior of supergeometry is the very key point to assure the consistency of all the calculations done by physicists.

\bigskip

Nowadays there seems to be a consensus on what supergeometry is, what is a superspace, a supermanifold or a supervariety and how to put geometric structures on them. The mathematical description has become quite clear. Our goal in this paper is not to redefine these concepts nor to reformulate supergeometry. Our more modest goal is to understand what is the nature of   {\it superfields}.  Often, mathematicians say that superfields are maps of super spacetime to a manifold (it could also be a supermanifold in some cases). We will see in a moment, with a simple example, why there is an issue with this definition. The problem was formulated for example in Ref. \cite{ar} in the Rogers--De Witt approach, so it is not new at all. It was mentioned also that in the sheaf-theoretical approach, one would need to appeal to auxiliary odd variables. We will analyze this statement in the present paper.

\bigskip

The definition of `even' or `bosonic' fields offers no difficulty.  They can be  functions on spacetime  valued in a finite dimensional manifold,  sections of a vector bundle, connections over it...  All these objects have a precise geometrical meaning and offer no ambiguity both, in their physics usage and in their mathematical formulation.

An `odd' or `fermionic' field is more difficult to interpret. In physics, it is vaguely stated that it is an `odd function on spacetime', that is, in the simplest case, fermionic fields are functions with values in some super vector space (in its odd part) or superalgebra. However, physicists use properties of odd fields that cannot be reconciled with this naive point of view.

Let us take the simplest example of an odd field on a spacetime $\R^n$ (we are not worrying here about Lorentz invariance), say  $\psi(x)$, $x\in \R^n$. It is sometimes stated that $\psi$ can be seen as a  function $f:\R^n\rightarrow \R$ multiplied by an odd parameter or odd auxiliary variable $\eta$, so $\psi(x)=f(x)\eta$.  More rigourously, this can be interpreted as a map between affine superspaces over $\R$
$$\begin{CD} A^{n|1}@>\psi >>A^{0|1}\,,\end{CD}$$ whose dual map on the corresponding superalgebras is given in terms of the global coordinates (see the Remark \ref{charttheorem-rem} about the Chart Theorem)
\beq \begin{CD}\wedge[\phi]@>\psi^{\sharp} >> \C^\infty(\R^n)\otimes \wedge[\eta]\\
\phi @>>>\psi^{\sharp}(\phi)=f(x)\eta\,.\end{CD} \label{naive}
\eeq
The symbol $\wedge[\phi]$ stands for the Grassmann or exterior algebra in the variable $\phi$.
In physics one needs to compute quantities such as
\beq \psi(x),\,\dots\quad ,\psi(x_1)\psi(x_2)\cdots \psi(x_m),\dots ,\frac{\partial \psi(x)}{\partial x^\mu},\,\dots\quad ,\frac{\partial^{n_1} \psi(x)}{\partial (x^{\mu_1})^{n_1}}\cdots
\frac{\partial^{n_k} \psi(x)}{\partial(x^{\mu_k})^{n_k}},\, \dots \label{physcal}\eeq
where the superindex labels the coordinates of spacetime and the subindex labels diferent points of spacetime.

These quantities, with products of any number of derivatives,  appear in Lagrangians and  other observables, and products of fields at different points of spacetime are the classical limit of correlation functions. They are always assumed to be, generically, different from zero.

It is clear that these calculations could not be reproduced with a simple object such as $\psi(x)=f(x)\eta$. All the quantities written above would be identically zero.  It seems that physicists and mathematicians, although not disagreeing on the nature of supermanifolds and supervarieties,  are not considering the same object when speaking about superfields or just odd fields.
Often this subtlety has been overlooked,  but we think that it deserves special attention. There must be an appropriate object, in the context of the mathematicians approach to supergeometry, that reproduces the properties of a superfield.

Not always supersymmetric theories are formulated in terms of superfields: for a high number of supersymmetries (that is, many odd coordinates in super spacetime) there appear too many component fields which would have to be constrained. Since it is not easy to find appropriate  constraints (one can even wonder if they exist), physicists devised other methods to implement supersymmetry in field theory. Nevertheless, whenever it is possible, superfield techniques are very powerful and highly desirable. In that case, the primordial objects from which everything else is derived, including the geometry of superspaces at play are superfields. Presumably,  one should be able to recover the original definition of supermanifold with its sheaf of superalgebras in terms of superfields. In Section \ref{evenrules-sec} we will demonstrate that this is possible with the use of the {\it even rules principle} explained in that same section.

In the second part of the paper, Sections \ref{algebraic-sec} and \ref{nonalgebraic-sec}, we deal with a variety of constraints imposed on the space of superfields. These are constraints that have been used in some supergravity inspired cosmological models (the system studied in Section \ref{newconstraint-sec} appeared for the first time in a preliminary version of this paper). To understand the meaning of these constraints was what got this work started. They provide a wonderful example on how one has to be careful when speaking about superfields. We will see that the use of {\it superschemes} is necessary to provide  an adequate mathematical framework for them and to elucidate its behaviour under supersymmetry transformations.

 Finally, in Section  \ref{observables-sec} we comment on the  interpretation of fermionic observables in the classical and quantum realms.

\bigskip

In the text we have tried to introduce the basic notions of algebraic geometry that are required to understand the generalization to the super setting. Some more basic concepts, as the definition or sheaf, etc are given in the Appendix \ref{basic-ap}. We have tried to give a consistent account of these concepts as a guide for the reader, but this paper is not a suitable place to introduce oneself to algebraic geometry, for which many good textbooks exist (particularly useful for us has been  Ref. \cite{eh}). For the super setting, introductory references are Refs. \cite{dm, va, fl} and a more detailed monograph is Ref. \cite{ccf}.
Physics conventions regarding spinor notation and supersymmetry transformations are given in Appendices \ref{spinors-ap} and \ref{susy-ap}.

\section{ Superspaces and scalar superfields}\label{sclarsf-sec}

\subsection {Some mathematical definitions} In this section we will consider the simplest case: an unconstrained, scalar superfield. The ground field will be $\R$

We need first some mathematic terminology. In Appendix \ref{basic-ap} we recall the standard definitions of {\it sheaf} over a topological space and of {\it morphism of sheaves}  (Definitions \ref{sheaf-def} and \ref{sheafmor-def}), which are used in what follows. The reader interested in a more complete and deep treatment of the subject in the  formalism that we use, can consult for example Refs.  \cite{dm,va,ccf,fl}.

 \begin{definition}\label{superspace-def}A {\it superspace}\footnote{The concept of superspace here is more general than the one used in physics, where usually restricts to Minkowski super spacetime.} $S=(|S|,\cO_S)$ is a topological space $|S|$ endowed with a sheaf of superalgebras $\cO_S$ such that the stalk at each point $x\in |S|$, denoted as $\cO_{S,x}$, is a  local superalgebra (it has a unique maximal ideal). The sheaf $\cO_S$ is the {\it structural sheaf} of the superspace $S$.

 \hfill$\square$
 \end{definition}

 The elements of $\cO_S(U)$, for $U\subset_{\mathrm{open}}|S|$ are {\it local sections over $U$}. If the open set is the total space $|S|$, then the  elements of $\cO_S(|S|)$ are called {\it  global sections}. The superalgebra $\cO_S(|S|)$ is also called the {\it coordinate superalgebra} of the superspace $S$, and  it is denoted simply as $\cO(S)$.

 \begin{example}\label{affinesuperspace}The {\it affine superspace} $A^{m|n}$ consists of the topological space $\R^m$  with the sheaf of superalgebras that, for any open set $U\in \R^{m}$, attaches the superalgebra
 $\cO^{m|n}(U):=C^\infty(U)\otimes \wedge[\theta^1, \dots ,\theta^n]\,.$

  \hfill$\square$
 \end{example}

Let  $ x^1, \dots, x^m$ be global coordinates on $\R^m$. Then we say that $ x^1, \dots, x^m, \\\theta^1,\dots ,\theta^n \in \cO(A^{m|n})$ are {\it global coordinates} on $A^{m|n}$.  Here, and as a general rule, Latin letters denote even (commuting) quantities and Greek letters denote odd (anticommuting) quantities.

 It is important to pay attention to the definition of morphisms of superspaces. One can define them as morphisms of the corresponding sheaves but first one has to put them over the same basis. This is done by using the pullback sheaf.

 \begin{definition}\label{morphism}
 A {\it morphism} $f:S\rightarrow T$  of superspaces is given by a pair $(|f|, f^\sharp:)$ where $|f|:|S|\rightarrow  |T|$ is a continuous function $f^\sharp:\cO_T\rightarrow |f|_*\cO_S$ is a morphism of sheaves that preserves the maximal ideal of the stalks. The sheaf $|f|_*\cO_S$ is the pullback by $|f|$ of the sheaf $\cO_S$, defined, for any open $U \subset T$, as
 $$|f|_*\cO_S(U)=\cO_S(|f|^{-1}(U))\,.$$

 \hfill$\square$
 \end{definition}

 \begin{remark}\label{charttheorem-rem}

 It is not difficult to show that a morphism  $f:A^{m|n}\rightarrow A^{p|q}$ is determined by the images of the global coordinates $(x^1, \dots, x^p;\theta^1,\dots ,\theta^q)$
  of $A^{p|q}$ under $f^\sharp$
  $$\begin{CD}
  \cO^{p|q}(\R^p) @>f^\sharp>>  \cO^{m|n}(\R^m)\\
  x^1, \dots , x^p@>>> f^1, \dots ,f^p\\
  \theta^1, \dots, \theta^q@>>> \eta^1,\dots ,\eta^q\,,
  \end{CD}$$

  \smallskip

 \noindent where $f^1, \dots ,f^p$ are even elements of $\cO^{m|n}(\R^m)$ and  $\eta^1,\dots ,\eta^q$ are odd elements of $ \cO^{m|n}(\R^m)$.
  This is the   {\it Chart Theorem} (see for example Theorem 4.1.11 in                            Ref \cite{ccf}).

  \hfill$\square$

 \end{remark}

 We will denote as $\rss$ the category of superspaces with their morphisms as defined above.

\begin{definition} A {\it supermanifold} of dimension $m|n$ is a superspace that is locally isomorphic to $A^{m|n}$.

  \hfill$\square$
 \end{definition}

 Of special importance for us will be the concept of  the {\it functor of points} associated to supermanifolds (and, later on, to superschemes).

 \begin{definition}{\sl Functor of points.} \label{functorofpoints-def}Let us denote as $\rsm$ the category of supermanifolds with morphisms the superspace morphisms such that, in the notation of Definition \ref{morphism}, $|f|:|S|\rightarrow  |T|$
 is a $C^\infty$ function. Let $S$ and $X$ be supermanifolds. A morphism $\varphi:S\rightarrow X$ is called an $S$-point of $X$.  The set of $S$-points of $X$ is the set $\Hom(S, X)$.

   The association
 $$\begin{CD} \rsm@>F_X>>\sets\\
 S@>>>\Hom(S, X)\end{CD}$$
 is a contravariant functor in $S$ called the functor of points of $X$.

 \hfill$\square$

 \end{definition}

 When $S$ is taken as the spectrum of the field $k=\R$ one is left with ordinary points of the topological space underlying the supermanifold. While these points are enough to recover the ordinary manifold, it is not so with a supermanifold, where, in principle, one would need  the $S$-points for any supermanifold $S$.

 If $Y$ is another supermanifold and we have a morphism $f:X\rightarrow Y$ then we can associate to each $S$-point of $X$, $\varphi: S\rightarrow X$, an $S$-point of $Y$ by composition, namely $f\circ \varphi$. So one can talk about the map $f$ as we do in the classical case,  as sending `points of $X$' to `points of $Y$'. It follows from Yoneda's lemma that it is equivalent to give the actual map $f$ as a map of ringed spaces or the natural transformation between the functors of points.

\subsection{Superfields}\label{superfields-sec}

In order to introduce the discussion on superfields we will first  describe heuristically  a toy model.

 \begin{paragraph}{\sl \bf Scalar superfields on $A^{1|1}$.} We consider the simplest model of superspace having both, even and odd components, the affine superspace over $\R$, $A^{1|1}=(\R, \cO^{1|1})$  with structural sheaf
 $$\cO^{1|1}(U)=C^\infty(U)\otimes \wedge[\theta],\qquad U\subset_{\mathrm{open}}\R\,.$$    The algebra of global sections will be denoted as $$\cO(A^{1|1}):= C^\infty(\R)\otimes \wedge[\theta]\,,$$
with $x$ and $\theta$  the global coordinates on $\R^{1|1}$. Generically,  a global section can be written as
 \beq \tilde\Phi=\tilde{A}+\tilde{G}\,\theta\,,\label{section}\eeq
 with $\tilde{A}, \tilde{G} \in C^\infty(\R)$.

Formally, (\ref{section}) looks like the superfields that appear in the physics literature, except for the fact that both, $\tilde{A}$ and $\tilde{G}$, are  ordinary (even) smooth functions on $\R$.

  Let us now consider $S_n=A^{0|n}$, a superspace   whose topological space is just a point and the structural sheaf is given by the Grassmann algebra over $\R$ in $n$ generators $\wedge[\xi^1,\dots ,\xi^n]$  (a {\it superpoint}).
  We consider a morphism
\beq  \begin{CD}S_n\times M@>\phi>>\R\,.\end{CD}\label{superpoint}\eeq
 Using the Chart Theorem \ref{charttheorem-rem} we can express it terms of the algebras and global coordinates as
  $$ \begin{CD}C^\infty(\R)@>\phi^\sharp>>\C^\infty(\R)\otimes \wedge[\theta]\otimes \wedge[\xi^1,\dots \xi^n]\\
  z@>>>\Phi:=\phi^\sharp(z)=A+\chi\,\theta \end{CD}$$ where  $A\in \bigl(\C^\infty(\R)\otimes \wedge[\xi^1,\dots \xi^n]\bigr)_0$ and $\chi\in \bigl(\C^\infty(\R)\otimes \wedge[\xi^1,\dots \xi^n]\bigr)_1$.

  \smallskip

 $A$ and $\chi$ are quantities that could represent the {\it component fields} of the {\it superfield} $\Phi$, being $A$ a bosonic (even) field and $\chi$ a fermionic (odd) field.  We are aware that we have introduced extra odd variables $\xi_1, \dots , \xi_n$ and we will argue why they are spurious, that is, they have no physical meaning.

  For the moment, let us see how this helps with the problem stated in the introduction. Let us consider $n=2$ and  compute some of the quantities in (\ref{physcal}) for the odd field $\chi$. We have
$$A(x)=A_0(x)+A_{12}(x)\xi^1\xi^2,\qquad \chi(x)= G_1(x)\xi^1+G_2(x)\xi^2\,,$$
with $A_0, A_{12}, G_1, G_2 \in \C^\infty(\R)$:

\begin{align*}&\chi(x)\chi(x')=(G_1(x)G_2(x')-G_2(x)G_1(x'))\xi^1\xi^2,\\ &\chi(x)\dot \chi(x)=(G_1(x)\dot G_2(x)-G_2(x)\dot G_1(x))\xi^1\xi^2\,,\end{align*} which, generically, are different from zero. These properties are used in classical field theory and have a (deformed, as $\hbar\neq 0$) counterpart in quantum field theory.

It is clear now that if one wants to have products of the same field in $n$ different points, or products of the field and its derivatives that are not identically 0, one just needs to increase the number of odd generators $\xi^i$ up to $n$. In order to achieve full generality, one is then lead to consider all the Grassmann algebras. This is reminiscent of the use of infinitely many variables in that DeWitt--Rogers approach, introduced to give a meaning to the  calculations done by physicists. It is seems clear from this example that one cannot escape from considering infinitely many odd variables that are non physical in some way.

\hfill$\blacksquare$

 \end{paragraph}

 The purpose of defining superfields is better served if one considers in (\ref{superpoint})  arbitrary, finite dimensional supermanifolds $S$ instead of only Grassmann algebras (or rather, the superpoints they define) and then require functoriality in $S$. This natural approach to superfields, that generalizes and improves our previous definition in terms of Grassmann algebras is not completely unknown to mathematicians, and it was was suggested to us by P. Deligne\footnote{ We are grateful to him for this precious input.} (see also Ref. \cite{bk}). However, in our experience,  it is not widespread knowdledge. Certainly it is ignored by the physics community and we think that tending a bridge between different ways of thinking is worthwhile. We will dedicate some time to explain it in detail.

On the other hand, and perhaps paradoxically in view of (\ref{physcal}), it is  useful to consider first the case where there is no dependence on the coordinates of spacetime: this means that spacetime is just a point and the  structural sheaf of $M$ is given, in the only open set (the point), by a certain finite dimensional superalgebra.

\begin{paragraph}{ \bf Spacetime is a point.}  We first need to introduce the concept of {\it adjoint} of a functor.

\begin{definition}Given two categories $\cC$ and $\cD$ and two functors
\begin{diagram}
\cC & \pile{\rTo^F\\  \lTo_G} & \cD\,,\\
\end{diagram}
we say that the functor $G$ is a right adjoint to the functor $F$ or that $F$ is a left adjoint to the functor $G$ if there is a bijection
\beq \Hom_D(F(X),Y)\cong\Hom_C(X,G(Y))\,,\label{adjointfunctor}\eeq
for all objects $X$ in $\cC$ and all objects $Y$ in $\cD$. Varying $X$ and $Y$ we have a family of bijections
$$\Hom_D(F(\,\_\,),\,\_\,)\cong\Hom_C(\,\_\,,G(\,\_\,))\,.$$

\hfill$\square$
\end{definition}

\begin{example}Let us consider the category (sets) of sets. Let $A$ be a fixed set and consider the functor that takes the Cartesian product by $A$
$$\begin{CD}
\sets@>F>>\sets\\
X@>>>X\times A\,.
\end{CD}$$

We have $$\Hom(X\times A, Y)\cong\Hom \bigl(X, \Hom(A, Y)\bigr)\,,$$
So the functor
$$\begin{CD}
\sets@>G>>\sets\\
Y@>>>\Hom(A, Y)\,
\end{CD}$$
is the right adjoint to to $F$.

\hfill$\square$
\end{example}

\begin{example}Let $A$ be a commutative superalgebra and let $\cC$ be the category of $A$-modules. $\cC$ has a tensor product defined in it \cite{dm}. For a certain $A$-module $N$ we consider the functor
$$\begin{CD}
\cC@>F>>\cC
\\
M@>>>M\otimes_A N\,.
\end{CD}$$
This functor has a right adjoint
$$\begin{CD}
\cC@>G>>\cC
\\
P@>>>\cHom_A(N, P)\,,
\end{CD}$$
where $\cHom_A(N, P)$ is the set of all $A$-linear maps $N\rightarrow P$. Notice that we have used a different notation than for the usual module morphisms $\Hom_A(N, P)$, which preserve parity.   The set $\cHom_A(N, P)$ is itself an $A$-module. Because of this, its elements are called {\it inner morphisms} (the set of inner morphisms is inside the category of $A$-modules $\cC$). The even part of $\cHom_A(N, P)$ are just the standard module morphisms.  The adjunction formula reads

$$\Hom_A(M\otimes_AN, P)=\Hom(M, \cHom(N,P))\,.$$

A very similar construction is the one that we will use for supermanifolds.

\hfill$\square$
\end{example}

The category of supermanifolds admits categorical products. We will not enter in the details here, which can be found, for example, in Ref. \cite{va}, page 137. Essentially, if $X$ and $Y$ are two supermanifolds, we have that  $|X\times Y|=|X|\times|Y|$, and the sheaf  can be determined by gluing the local tensor products $\cO_X(U)\hat\otimes \cO_Y(V)$, where $U$ and $V$ are rectangular open sets admitting global coordinates and $\hat \otimes$ is the completion of the tensor product.

The cases that we will treat here, spacetime being a point, are specially easy in that respect.

So, let $M$ be a superspace with underlying topological space a point. We define the functor $F$ as
\beq\begin{CD}\rsm@>F>>\rsm\\
S@>>>S\times M
\,.\label{adjoint}\end{CD}\eeq  $F$ just  takes the cartesian product by $M$.

 This functor has a right adjoint since $\rsm$ is a closed monoidal category \cite{bk}. The adjoint functor associates to every supermanifold $Y$ the supermanifold of  {\it inner morphisms} from $M$ to $Y$, denoted as  $\cHom(M, Y)$
$$\begin{CD}\rsm@>G>>\rsm\\
Y@>>>\cHom(M, Y)
\,,\label{adjoint2}\end{CD}$$
which is defined through the formula (\ref{adjointfunctor})
$$\Hom(S\times M,Y)\cong\Hom(S,\cHom(M, Y))\,.$$

The functor $S\rightarrow \Hom(S\times M,Y)$,  that we were considering heuristically in (\ref{superpoint}) is in fact the functor of points (Definition \ref{functorofpoints-def}) of the {\it supermanifold of inner morphisms} $\cHom(M, Y)$.

It is more clear now why we have chosen $M$ to be of such restricted type first. For $M$ an arbitrary supermanifold, a superspace like $\cHom(M, Y)$ would be infinite dimensional, which is considerably more difficult to describe. In \cite{bk}, for example, infinite dimensional superspaces of maps between supermanifolds are studied in depth, but we will stay with this simpler description. We will come back to this issue later.

\medskip

It is instructive to compute it for a particularly simple example.

 \begin{example} Let us set $M=A^{0|1}$ with coordinate $\theta$ and  and $Y=\R$ with coordinate $z$. A morphism $S\times A^{0|1}\rightarrow \R$ is given in terms of the global coordinates as
 \beq\begin{CD}\cO(\R)@>>>\cO(S)\otimes \cO^{0|1}(\R)\\
 z@>>> f+\eta\theta, \end{CD} \qquad f\in \cO(S)_0,\; \eta\in \cO(S)_1\,.\label{family0}
 \eeq

It is then clear that $\cHom(M, \R)\cong A^{1|1}$, since any morphism $S\rightarrow A^{1|1}$ is given in terms of the global coordinates as
$$\begin{CD} \cO^{1|1}(\R)@>>>\cO(S)\\
t, \,\tau @>>> f, \,\eta  \,.\end{CD}$$

It amounts to the same to give a map $S\times A^{0|1}\rightarrow \R$ than to give a map $S\rightarrow A^{1|1}$.

\hfill$\square$
 \end{example}

More generally, we can take $M=A^{0|q}$, with coordinates $\theta_1,\theta_2,\dots, \theta_q$ and target manifold $\R$. A morphism $S\times A^{0|q} \rightarrow \R$ is given by

$$\begin{CD}\cO(\R)@>>>\cO(S)\otimes \cO^{0|q}(\R)\\
 z@>>> \Phi  \,,\end{CD}
 $$
 with
\beq\Phi= \Phi_0+\Phi_i\theta^i+\Phi_{ij}\theta^i\theta^j+ \Phi_{ijk}\theta^i\theta^j\theta^k+\cdots,\label{expansion}\eeq
(sum over repeated indices is understood). We can restrict the sum to  $i<j<k<\cdots$ due to the antisymmetry of $\theta^i\theta^j\theta^k\cdots$. The component fields are
\beq \Phi_0, \Phi_{ij},\dots \in \cO(S)_0,\qquad \Phi_i, \Phi_{ijk},\dots \in \cO(S)_1\,,\eeq
and $S$ runs over all supermanifolds.
 Then $\cHom(A^{0|q}, \R)\cong A^{2^{q-1}|2^{q-1}}$.

 \medskip
 We observe that the coordinates in $S$  play the role of {\it parameters}. In order to see this, it is useful to take  first an $S$  that is an ordinary manifold. Then, its coordinates are all even. Choosing  a value for them,  the expression (\ref{expansion}) produces a particular inner morphism $M\rightarrow \R$ which, in this case,  would be a supermanifold morphism. In order to obtain the rest of the inner morphisms is necessary to introduce an  $S$ with odd coordinates. We then generalize the concept of `parameter' to include also the odd coordinates. We  say that  (\ref{expansion}) is a {\it family of morphisms $M\rightarrow \R$ parametrized by $S$} \cite{de}.

\smallskip

Let us now denote $\cH=\cHom(M, \R)\cong A^{2^{q-1}|2^{q-1}}$, the {\it  superspace of superfields with target $\R$}. We can consider the family (\ref{expansion}) also over  $\cH$, that is, the family $\cH\times A^{0|q}\rightarrow \R$
\beq \begin{CD}\cO(\R)@>>>\cO(\cH)\otimes \cO^{0|q}(\R)\end{CD}\label{universal}\,.\eeq
In (\ref{expansion}), the coefficients $\Phi_{ijk\cdots}$ can be taken as coordinates of $\cH$.
This family is universal in the sense that, any family of maps (\ref{expansion}) parametrized by $S$ is determined, up to isomorphism, by a map $S\rightarrow \cH$.

We may observe here a similarity with the naive approach (\ref{naive}), where for each of the odd components $\Phi_{ijk\cdots}$  of the superfield, one would introduce an auxiliary odd variable. With this interpretation, in terms of the adjoint functor, we see that the interesting object (the universal one) has also even coordinates, one for each even $\Phi_{ijk\cdots}$.

 \hfill$\square$

  \bigskip

In summary, to speak about superfields we have to speak about the superspace of such superfields, and we have clarified how to define it in terms of its functor of points.
The undetermined number of `auxiliary variables' in other approaches are just coordinates in $S$, and they appear because of the functorial description.

In the toy model considered (spacetime is a point) one can explicitly give the superspace of superfields in terms of sheaves  as in Definition
\ref{superspace-def}. One could then, in principle, avoid the use of the functor of points. Then the  `auxiliary variables'
 ($S$) would not appear, so they cannot have any physical meaning attached.

What we find remarkable  is that the object that physicists have been using all the time, without giving an explanation, was exactly the functor of points
of the superspace of inner morphisms. This is just another example of how the language of the functor of points converges with the physicist's point
of view of supergeometry\footnote{Another clear example is the definition of supergroups, which in all references is given explicitly in terms of the
functor of points without mentioning it.}.
 We hope to contribute with these simple examples to the understanding of the concept of superfield.

\hfill$\blacksquare$

\end{paragraph}

If spacetime is not a point and $M$ is a supermanifold modelled on $A^{p|q}$, the space of inner morphisms becomes infinite dimensional and we may have trouble describing it. We can, though, preserve the idea of a family of maps   $M\rightarrow \R$ parametrized by a supermanifold $S$. This presents no interpretation problem. For $M$ an affine space $A^{p|q}$ we will have an expansion like (\ref{expansion}), where now all the component fields will acquire a dependence on the even coordinates of $M$, $(x^1, \dots, x^p)$. Then, the very definition of superfield in terms of parametrized families takes care of all the properties that physicists use. One can construct Lagrangians, functional actions and use the calculus of variations to determine the field equations, very much in the same way that they are computed in ordinary field theory.

\bigskip

Since in several cases the constraints that we are going to investigate are algebraic, that is, they do not involve spacetime derivatives,
the dependence on the even coordinates can be factored out and everything works as if  spacetime would shrink to a point. Even there, the definition of the superfield that we give will make appear non trivial solutions of the constraints, with a little difficulty added: the superspaces that we will have to consider are not always smooth supermanifolds nor algebraic varieties, and we will have to allow for singular {\it superschemes}

\section{The  even rules principle}\label{evenrules-sec}

In  this section we explain the {\it even rules principle}, a result due to Deligne and Morgan \cite{dm}, and see its implications for superfields.

 We set $k=\R,\C$. The algebras and superalgebras that we consider here are $k$-algebras, unless otherwise stated, and always  have unit.

 We will denote by $\svs$ the category of  super vector spaces  and by $\csa$  the category of commutative (or {\it supercommutative}) superalgebras with unit, both over $k$.

Let $\cA$ be a commutative superalgebra and $M_\cA$ a left $\cA$-module, that is, $M_\cA$ is a super vector space with a morphism of super vector spaces $$\begin{CD}\cA\otimes M_\cA@>>>M_\cA\\a\otimes m@>>>a\cdot m\,,\end{CD}$$ satisfying
$$a\cdot (b\cdot m)=(ab)\cdot m,\qquad a,b\in \cA,\quad m\in  M_\cA\,.$$ Left $\cA$-modules for commutative superalgebras are also right $\cA$-modules by setting
$$\begin{CD}M_\cA\otimes \cA@>>>M_\cA\\m\otimes a@>>>m\cdot a:= (-1)^{p(m)p(a)}a\cdot m\,,\end{CD}$$
where $p(m)$ and $p(a)$ are the parities of $m$ and $a$ respectively. This  satisfies
$$(m\cdot b)\cdot a=m\cdot (b\cdot a),\qquad a,b\in \cA,\quad m\in  M_\cA\,,$$
We will just call them modules.

Let $V$ be an object in $\svs$  and $\cB$ and object in $\csa$.  We denote as $V(\cB):=\cB\otimes V$ the extension of the scalars of $V$ by $\cB$ (see Definition \ref{extensions-def}).
Let $h:\cB\rightarrow \cB'$ be  a morphism of commutative  superalgebras. Then $V(\cB')$ is also a $\cB$-module by further extending the scalars to $\cB'$. There is a morphism  of $\cB$-modules
$$\begin{CD} V(\cB)@>V(h)>>V(\cB')\\
 b\otimes v@>>> h(b)\otimes v\,.\end{CD}$$ This is well defined since $h$ is a superalgebra morphism.

We can as well take the even part of the module.
 If $V=V_0+V_1$ and $\cB=\cB_0+\cB_1$ are the splitting in even and odd parts, we denote
 $V(\cB)_0$ the even part of $V(\cB)$,
  $$V(\cB)_0=\cB_0\otimes V_0\oplus \cB_1\otimes V_1\,,$$
  which is a $\cB_0$-module. Given a morphism of superalgebras $h:\cB\rightarrow \cB'$, then $V(h)_0$ is a morphism of $\cB_0$-modules
$$\begin{CD} V(\cB)_{0}@>V(h)_0>>V(\cB')_0\\
 b\otimes v@>>> h(b)\otimes v\,.\end{CD}\label{extmor}$$

 \medskip

\begin{remark}\label{diagrams} It will be useful here to recall the definition of superalgebra in terms of commutative diagrams.

Let $V$ be a super vector space.  An associative superalgebra structure on $V$ is given by a linear map, {\it the product}:
 $$\begin{CD}V\otimes V@>\pi>>V\end{CD}$$
 such that
 $$p(\left(\pi(u\otimes v)\right)=p(u)+p(v)$$
and satisfying
 the  associativity property, that is, that the diagram
$$
\begin{diagram}[notextflow]
                    &                      &V\otimes V &        & \\
                    & \ruTo^{\pi\otimes\id} &           &\rdTo^\pi & \\
V\otimes V\otimes V &                      &           &        & V\\
                    & \rdTo_{\id\otimes \pi}   &           &\ruTo_\pi & \\
                    &                      &V\otimes V &        &
\end{diagram}$$
is commutative. On the other hand, we say that the superalgebra is commutative if, given the {\it flip map},
$$\begin{CD} V\otimes V@>c_{V,V}>>V\otimes V\\
v\otimes w@>>>(-1)^{p_vp_w}w\otimes v,\end{CD}$$
the diagram

$$
\begin{diagram}[notextflow]
             &                &V\otimes V \\
             &\ruTo^{c_{V,V}} &  \dTo_\pi    \\
V\otimes V   &  \rTo_\pi              &     V
\end{diagram}$$
commutes.

\hfill$\square$
\label{commdiag}
\end{remark}

 If $V$ has a $k$-superalgebra structure, then $V(\cB)$ has a $\cB$-algebra structure naturally: Let $\pi:V\otimes V\rightarrow V$ be the product on $V$, then
 we have a product\footnote{We use here the direct product and not the tensor product because the bilinearity is in $\cB_0$ and not only in $k$.}
 \beq \begin{CD}V(\cB)\times V(\cB)@>\pi_\cB>> V(\cB)\end{CD}\label{products}\eeq given simply by
 $$ \pi_\cB\left(b_1\otimes v_1, b_2\otimes v_2\right)=(-1)^{p(v_1)p(b_2)} b_1b_2\,\pi(v_1\otimes v_2)\,.$$  It is straightforward to check that the associativity and commutativity properties are satisfied.

\begin{definition}\label{functorialprop} Let $V$ and $W$ be two superspaces. We say that a family of morphisms $$\bigl\{\,f_\cB: V(\cB)\rightarrow W(\cB),\quad \cB\in\csa\,\bigr\}$$ is {\it functorial} in $\cB$ if
 given a superalgebra morphism $$\begin{CD}\cB@>h>> \cB'\end{CD} $$ the diagram
$$\begin{CD}
V(\cB)@>f_{\cB}>>W(\cB)\\
@V{ V(h)} VV@V V{W(h)} V\\
V(\cB')@>f_{\cB'}>>W(\cB')\,.
\end{CD}$$
commutes.
\hfill$\square$
\end{definition}
In particular, if $V$ has a superalgebra structure, it is not difficult to see that the family $\{\,\pi_\cB:V(\cB)\otimes V(\cB)\rightarrow V(\cB)\,\}$ is functorial in $\cB$. The same is true for the families of maps appearing in the associativity and commutativity diagrams for the algebras $V(\cB)$.

\medskip

The following theorem of Deligne and Morgan (Ref.\cite{dm}, page 56) will allow us to recover the defining sheaf in terms of scalar superfields.

\begin{theorem} \label{erp}{\sl Even rules principle.} Let $\{V_i\}_{i\in I}$, $I=1, \dots , n$ be a family of super vector spaces, $V$ another super vector space and  $\cB=\cB_0\oplus \cB_1$ a commutative superalgebra. As before, we denote
$ V_i(\cB)_0=(\cB\otimes V_i)_0$ and $V(\cB)_0=(\cB\otimes V)_0$.

Any family of $\cB_0$-multilinear maps
$$\begin{CD}V_{1}(\cB)_0\times \cdots \times V_{n}(\cB)_0@>f_\cB>> V(\cB)_0\end{CD}$$
which is {\it functorial} in $\cB$ comes from a unique  morphism
$$\begin{CD}V_1\otimes \cdots\otimes V_n@>f>>  V\end{CD}$$ as in (\ref{products}), that is,
$$f_\cB(b_1\otimes v_1, b_2\otimes v_2,\dots, b_n\otimes v_n)=(-1)^pb_1\cdots b_n\, f(v_1\otimes\cdots \otimes v_n)\,, $$ where $p$ is the number of pairs $(i,j)$ with $i<j$ and $v_i, v_j$ odd.

\end{theorem}

\begin{proof}

 We will not prove the theorem here (see Ref. \cite{dm}), but it is  instructive to see how the map $f$ can be recovered from the family of maps $f_\cB$. Let us consider the simple case of a family of maps
$$\begin{CD}V(\cB)_0\otimes V(\cB)_0@>f_\cB>>V(\cB)_0\,,\end{CD}$$ then we have three possible cases:
\begin{enumerate}
\item $v_1, v_2$ even. Then we may take $b_1=b_2=1$ (in an arbitrary algebra $\cB$) and
$$f(v_1\otimes v_2):=f_\cB(v_1\otimes v_2)\,.$$
\item $v_1$ even, $v_2$ odd. Then we take  for example $\cB=\Lambda[\xi]$, $b_1=1$ and $b_2=\xi$. The equality
    $$f_\cB(v_1\otimes \xi v_2)=\xi f(v_1\otimes v_2)$$ determines $f(v_1\otimes v_2)$.
\item $v_1$ odd, $v_2$ odd. It is enough to consider $\cB=\Lambda[\xi^1, \xi^2]$ and the equality
$$f_\cB(\xi^1v_1\otimes \xi^2 v_2)=-\xi^1\xi^2 f(v_1\otimes v_2)$$ determines $f(v_1\otimes v_2)$.
\end{enumerate}
\end{proof}

One way  to give a structure of superalgebra to $V$ is to give a $\cB_0$-algebra structure on $V(\cB)_0$ which varies functorially with $\cB$. From the commutative diagrams in Remark \ref{commdiag}, it is clear that the superalgebra will be associative or commutative if the algebra structures on $V(\cB)_0$ are so.

\begin{example}{\sl Toy model.} \label{toymodel0}Let us consider the vector space over $k=\R, \C$ with basis one even vector $e$ and one odd vector $\theta$.
$$V^{1|1}=\rspan_k \{e, \theta\}\,.$$ Let $\cB$ be a commutative superalgebra. An element of $V_0^{1|1}(\cB)$ will be of the form
$$\Phi_\cB=b_0\otimes e + b_1\otimes\theta,\qquad b_0\in \cB_0,\quad b_1\in\cB_1\,.$$ Let $\Psi_\cB=a_0\otimes e + ba_1\otimes\theta$ another element of $V^{1|1}(\cB)_0$. We can define a product on $V^{1|1}(\cB)_0$ as
$$\Psi_\cB\,\bullet_\cB\,\Phi_\cB=b_0a_0\otimes e + (b_0 a_1+a_1b_0)\otimes\theta\,.$$ It is immediate to check the functorial property (Definition \ref{functorialprop}) of the whole family of products  $\{\,\bullet_\cB\,\}$: for a morphism $h:\cB\rightarrow\cB'$, we have the map
$$\begin{CD}V^{1|1}(\cB)_0@>V^{1|1}(h)_0>> V^{1|1}(\cB')_0\end{CD}\,,$$ that, for short and without risk of confusion, we will call simply $h$
\begin{align*}&h(\Psi_\cB\,\bullet_\cB\,\Phi_\cB)=h(b_0a_0)\otimes e+h(b_0 a_1+a_1b_0)\otimes\theta=\\&h(b_0)h(a_0)\otimes e+\left(h(b_0)h( a_1)+h(a_1)h(b_0)\right)\otimes\theta=\\&
h\bigl(\Phi_{\cB}\bigr)\,\bullet_{\cB'}\,h\bigl(\Psi_{\cB}\bigr)\,.
\end{align*}
One can also check the functoriality for the associativity and commutativity diagrams (Remark \ref{diagrams}). The algebra structure defined in the superspace $V^{1|1}$ converts it into the Grassmann algebra in one variable $\wedge[\theta]$.

\hfill$\square$
\end{example}

\begin{remark}\label{Grassmann-remark} \cite{dm} The same result can be obtained if, instead of considering arbitrary commutative superalgebras, one restricts to the subcategory of Grassmann algebras $\cB=\wedge[\xi^1,\dots, \xi^n]$, for short $\wedge^n$,  for $n$ arbitrarily large, but finite.
For $V$ a superalgebra (as in (\ref{toymodel0})), the object $(\wedge^n\otimes V)_0$ (without specifying  $n$), appears in the literature as the  {\it Grassmann envelope} of the superalgebra $V$.

\bigskip

\hfill$\square$
\end{remark}

Let $M$ be, as above, super spacetime. We consider   two  families of maps parametrized by $S$,
$S\times M\rightarrow\C$, which in terms of the algebras read
$$\begin{CD}\cO(\C)@>>>\cO(S)\otimes\cO(M)\\
z@>>>\Phi\end{CD}\qquad \begin{CD}\cO(\C)@>>>\cO(S)\otimes\cO(M)\\
z@>>>\Psi\end{CD}$$
where $\Phi$ and $\Psi$ are even sections of $\cO(S)\otimes\cO(M)$. The section given by the product  $\Phi\cdot \Psi$ defines another such family. By identifying $\cO(S)=\cB$ and $\cO(M)=V$, and applying  Theorem \ref{erp}, we see that it is the same giving the product in $V$ than giving the functorial family of maps $$\begin{CD}V(\cB)_0\otimes V(\cB)_0@>f_\cB>> V(\cB)_0\\
\Phi\otimes \Psi@>>>\Phi\cdot \Psi\,.\end{CD}$$
In physics,  it is the second option that is chosen and it is called the product of superfields,  although the nature of $\cB$ and the functoriality of $f_\cB$ is not clarified.

\section{How to deal with algebraic constraints.}\label{algebraic-sec}

The examples that we would like to analyze  are the result of some algebraic constraints imposed on a certain set of $N=1$, $D=4$ chiral superfields. In physics, one needs to consider the complexification of the affine space $A^{4|2}$ over the reals. This is because the action of the Lorentz group: the two odd components form a Weyl spinor and this representation is complex. In order to obtain a real form one needs to consider  $A^{4|4}$ over the complex numbers and then one can impose a reality condition compatible with the action of the Lorentz group. Two of the odd variables transform as a chiral (Weyl) spinor and the other two as an antichiral spinor.
So, in what follows we will consider complex, affine  superalgebras. This does not mean that we are considering algebras of holomorphic functions: we will consider all the real analytic ones, so we see $\C^4\cong \R^8$.

The super spacetime under consideration is  $M=A^{4|2}$, the {\it chiral superspace}, with global sections $$\cO(A^{4|2})=C^\infty(\C^4)\otimes \wedge[\theta^1,\theta^2]\,.$$

A {\it chiral superfield} associates to each supermanifold $S$ a family of morphisms parametrized by $S$
$$\begin{CD}S\times M@>>> \C\,,\end{CD}$$ given in terms of an even section of $\cO(S)\otimes\cO(M)$ as in (\ref{expansion})
\beq\Phi=A+\theta^\alpha{\chi}_\alpha + \theta^\alpha\theta_\alpha F,\qquad \alpha=1,2\,.\label{chiralsuperfield}\eeq
 The index notation is the usual in physics  and it is explained in Appendix \ref{spinors-ap}; in particular, sum over repeated indices is understood.  The component fields $A, F$ and $\chi_\alpha$ depend on the even coordinates of the chiral superspace, although we do not write it explicitly.

We want now to study  algebraic constraints on chiral superfields. Spacetime dependence is untouched by the constraints that we will consider, so one can  effectively consider that spacetime is reduced to a point: the full picture is recovered by tensoring with $C^\infty(\C^4)$ when needed.

If spacetime is a point, then the superspace of superfields, that is, the space of inner morphisms $M\rightarrow \C$ can be identified with $\cH=A^{2|2}$, being the even global coordinates $A$ and $F$ in (\ref{chiralsuperfield}) and $\chi_\alpha$, $\alpha=1,2$ the odd ones. The constraints will give relations among these coordinates. The restricted space is not necessarily an afine superspace or a supermanifold. In the cases we will deal with it is in fact an {\it affine superscheme}. In the next subsection we try to give a brief summary on the principal results on schemes and superschemes. In the non super case a complete treatment can be found in  any textbook on algebraic geometry (see for example the first chapter of Ref. \cite{eh}). For the  super case there is a thorough treatment in Ref. \cite{ccf}.

\subsection{Schemes and superschemes}\label{superschemes-sec}

An affine  algebraic variety $Y\subset \C^n$ is  the zero locus of some polynomials. One defines the algebra $F$ of polynomials over $Y$ as the restriction of the the polynomials $\C[x^1,\dots , x^n]$ to Y. In fact, $F$ can be viewed as  the quotient
$$F=\C[x^1,\dots , x^n]/I$$
where $I$ is the ideal of functions that vanish on $Y$, so two polynomials are considered equivalent if they differ by a polynomial that is zero on $Y$. Each point of $Y$ has associated a maximal ideal of $F$. In general, an ideal of $F$ corresponds to an algebraic subset (zero locus in $Y$ of sets of  polynomials in $F$).  The {\it Zariski topology} on $Y$ is such that the closed subsets of $Y$ are precisely the algebraic subsets. Then the open subsets are the complements of closed subsets.

We are going to define now the {\it structural sheaf} of $Y$, $\cO_Y$ (see Appendix \ref{sheaf-def}). We have to  associate to each open set in $Y$ a ring. Let  $U\subset Y$ be an open subset of $Y$. We define
$$\cO_Y(U)=\left\{\frac{f}{g}\; |\; f, g\in A, \, g(x)\neq 0 \; \;\forall x\in U\right\}\,.$$
This procedure is called the {\it localization} of $F$ at $U$. One can check that it satisfies all the conditions to be a sheaf. The stalk at a point $x\in U$ with maximal ideal $\fp$ is just
is the algebra
$${F}_\fp:=\left\{\frac f g\;\; \big|\;\; f\in F,\; g\in  F- \fp\right\}$$
 and it is a  {\it local algebra}, that is, it has a unique maximal ideal. One recovers the algebra $A$ as the algebra of global sections $\cO_Y(Y)=F$.  We will use interchangeably the notation $\cO_Y$ and $\cO_F$.

 \smallskip

 It is convenient to add some points to the topological space $Y$ and consider also all the irreducible subvarieties of $Y$. This is equivalent to consider, not only the maximal ideals but all the prime ideals\footnote{The  reason is that the preimage under an algebra morphisms of a prime ideal is a prime ideal, while this is not true for maximal ideals.}. The topological space that results is called the {\it spectrum} of $A$ and it is denoted as $\rspec(F)$. The Zariski topology is still well defined. The pair $\rspec(F)$ together with its structural sheaf will be denoted as
$$\uspec(F):=(\rspec(F), \cO_F)=(|Y|, \cO_Y)\,,$$
where $|Y|$ denotes $Y$ understood as a topological space.

\smallskip
Given a field $k$, an affine algebra is a finitely generated, noetherian, associative, commutative, unital $k$-algebra without nilpotent elements.
The procedure that we have described above establishes an equivalence of categories between affine algebras and affine varieties.

One can apply the same procedure to an arbitrary commutative algebra $F$. It can, for example, contain nilpotents: $\rspec(F)$ and $\cO_F$ still make sense. The nilpotent elements of a commutative  algebra form an ideal $N$ so one can define the {\it reduced algebra}
$F_\red:= F/N$.
Since $N$ sits inside every prime ideal, we have that
$\rspec(F_\red)\cong\rspec(F)$ as topological spaces. Nevertheless, the sheaves $\cO_F$ and $\cO_{F_\red}$ will be, in general, diferent. In particular, $F_\red$ may be an affine algebra even if $F$ is not.

The following example will be used later on.

\begin{example}\label{nonaffine-ex}{\sl Ring of the dual numbers} Let us consider the algebra of polynomials in one even  variable, $\C[\epsilon]$, and let  $(\epsilon^2)$ be the ideal generated by the element $\epsilon^2$. The quotient $F=\C[\epsilon]/( \epsilon^2)$ is not  an affine algebra, since it contains a nilpotent, namely, the element $\epsilon$. A generic element of $F$ will be of the form
  $$f_0 +\epsilon\,f_1,\qquad f_0, f_1\in \C \,.$$
 The solution of the polynomial equation $\epsilon^2=0$ over $\C$ is just $\epsilon=0$ and in fact, $F_\red\cong \C$. Nevertheless, the algebra $F$ keeps track of the double multiplicity of the solution, so it has more information. The only prime ideal is $\fp=(\epsilon)$  and the stalk of the sheaf at such point is
$$F_\fp=\left\{\;\frac {f_0+\epsilon\,f_1}{g_0+\epsilon\,g_1 }\quad \big|\quad f_i, g_i\in \C, \quad g_0\neq 0 \;\right\}\,.$$
Working with the reduced algebra, there is only one point in the spectrum, $(0)$, and the stalk at that point is simply $\C$.

\hfill$\square$
\end{example}

We are led to the following definition:

\begin{definition} An {\it affine scheme}\footnote{It may seem odd that the category of affine schemes relates to non affine algebras, but  the adjective `affine' on the noun `scheme' is used to distinguish it from a {\it projective scheme}, a generalization that we will not need in this paper.} $X$ is a topological space $|X|$ together with a sheaf of algebras $\cO_X$ which is isomorphic to $\uspec(F)$ for some algebra $F$.

\hfill$\square$

\end{definition}

 Since superalgebras inevitably contain nilpotents, the concept of scheme seems suitable for extension to superalgebras. A superalgebra $\cA=\cA_0+\cA_1$ is an {\it affine superalgebra} if its even part $\cA_0$ is finitely generated as an algebra, its odd part $\cA_1$ is finitely generated as an $\cA_0$-module and taking quotient by the ideal generated by the  odd nilpotents  $\cJ$,  $\tilde \cA:=\cA/\cJ$, one obtains an affine algebra, so $\cA$ contains no further nilpotents.
For any superalgebra, we will denote as $\cA_\red$ the superalgebra modulo the ideal generated by all the nilpotents.
 One can also construct the topological space $\rspec(\cA)=\rspec(\cA/\cJ)$, and equip it with a sheaf of superalgebras obtained by localization. We then have:

 \begin{definition}

An {\it affine superscheme} is a superspace $S=(|S|, \cO_S)$ which is isomorphic to $\uspec(\cA)$ for a superalgebra $\cA$ not necessarily affine.

\hfill$\square$

\end{definition}

Given an affine superscheme $S$ with superalgebra $\cA$ there is always an affine scheme $S_\red$  associated to the reduced algebra $\cA_\red$. It is the {\it reduced scheme} of the superscheme, a concept which is similar to the concept of reduced manifold of a supermanifold or reduced algebraic variety of an algebraic supervariety.

\subsection {Constraint $\Phi^2=0$}\label{phi2-sec}

We want to consider the chiral superfield (\ref{chiralsuperfield})
satisfying $\Phi^2=0$. We recall here that we are considering space time shrunk to a point. This means
$$\Phi^2=A^2+2A\theta^\al\chi_\al+(2AF-\frac 12\chi^\al\chi_\al)\theta^\be\theta_\be=0\,,$$ where $A, F\in \left(\cO(S))\right)_0 $ and $\chi_\al\in\left(\cO(S)\right)_1$.

The solutions to the constraint are indeed the superspace of inner morphisms  from $M$ to the spectrum of the ring of dual numbers (\ref{nonaffine-ex})
$$\begin{CD}M@>>>\uspec\left(\C[\epsilon]/(\epsilon^2)\right)\,.\end{CD}$$
This spectrum is not a supermanifold, but a superscheme, so we have to allow families of morphisms parametrized by superschemes $S$ --instead of simply supermanifolds manifolds or super algebraic varieties-- in order to describe the space of inner morphisms:
$$\begin{CD}S\times M@>>>\uspec\left(\C[\epsilon]/(\epsilon^2)\right)\,.\end{CD}$$

We obtain then a system of equations restricting  the coordinates of $\cH=A^{2|2}$:
\beq A^2=0,\qquad A\chi_\alpha=0,\qquad 4AF-\chi^\al\chi_\al=0\,.\label{phi2constraint}\eeq
They define an affine superscheme that we denote  as  $\cL$.
The superalgebra of global sections, $\cO(\cL)$  is
$$\cO(\cL)=C^\infty(\C^2)[ \chi_1,\chi_2]\big/\bigl( A^2,\, A  \chi_\alpha,\, 2AF-\chi_1\chi_2\bigr)\,.
$$

In order to study this scheme we start considering the algebra obtained by quotienting  $\cO(\cL)$ by the odd nilpotents, that is, putting the odd coordinates to zero.
The  affine scheme that we obtain  $\tilde\cL$ satisfies the quadratic relations
\beq A^2=0,\qquad AF=0\,,\label{reducedconstraints}\eeq
so the algebra defining the affine scheme is
$$\cO(\tilde\cL)=C^\infty(\C^2)\big/\bigl( A^2,\, AF\bigr),\qquad (A,F)\in \C^2\,.$$  Since $A$ is an even nilpotent, the scheme is not a supermanifold (nor algebraic variety).

There is, however,  an open set where the   scheme is isomorphic to an affine space. It corresponds to the points where $F$ is invertible, that is, to the prime ideals of $\cO(\tilde \cL)$ that do not contain $F$. We will denote the localization of $\cO(\tilde \cL)$ at these points as $\cO(\tilde \cL)_{F\neq 0}$. This essentially means that we can set $A=0$ from the second equation in (\ref{reducedconstraints}); then, the first one is satisfied identically:
\beq\cO(\tilde \cL)_{F\neq 0}\simeq C^\infty(\C^\times)\,.\label{redlocscheme}\eeq
This is the regular or smooth part of the scheme, represented by the object $\C^\times=\C-\{0\}$.

Going back to  the superscheme, we can do a change of variables for the coordinates:
\beq A'=4AF- \chi^\al\chi_\al,\qquad F'=F \qquad \chi_\al'=\chi_\al\,,\label{changevariable}\eeq with Jacobian
$$J=\begin{pmatrix}4F&4A&-2\chi_2&+2\chi_1\\
0&1&0&0\\
0&0&1&0\\
0&0&0&1\end{pmatrix}\,,$$
which is non singular for $F\neq 0$. The  inverse transformation is
$$ A=\frac{A'+\chi^\al\chi_\al}{4F}, \qquad \qquad F=F',\qquad \chi_\al=\chi_\al'\,,$$ and the smooth part of the superscheme is just given by the ring
$$C^\infty(\C^2)[\chi_1,\chi_2]_{F\neq 0}\big/ ( A')\simeq C^\infty(\C^\times)[\chi_1,\chi_2]\,.$$
This model appeared in its non linear form in Ref. \cite{vak}, so it is called the {\it Volkov-Akulov multiplet}. In terms of superfields it appeared in   Refs. \cite{ro, cdcg,ks}.

\subsection{General constraint $f(\Phi)=0$}
Let $f$ be a polynomial in one variable. We consider now the more general constraint
$$f(\Phi)=0\,.$$ Because of the nilpotency of $\theta^\alpha$, this reduces to
$$f(\Phi)=f(A)+\theta^\al\chi_\al f'(A)+\
\theta^\al\theta_\al\bigl(f'(A)F-
\frac 14f''(A)\chi^\al\chi_\al\bigr)\,.$$
For example, let us take $f(\Phi)=\Phi^n$. Then the constraints that define the affine superscheme $\cL$ are:
\beq A^n=0,\qquad A^{n-1}\chi_\al=0,\qquad
 A^{n-2}\left(4AF-(n-1)\chi^\al\chi_\al\right)=0\,.\label{generalconstraint}
\eeq
The reduced affine scheme in this case is defined  by the ring
$$\cO(\tilde \cL)=C^\infty(\C^2)\big/ \bigl( A^n,\, A^{n-1}F\bigr),\qquad (A,F)\in \C^2\,.$$
We can still localize at $F$ invertible and we get
$$\cO(\tilde \cL)_{F\neq 0}=C^\infty(\C^2)_{F\neq 0}\big/ ( A^{n-1} )\simeq C^\infty (\C\times \C^\times)\big/ \bigl( A^{n-1}\bigr)\,.$$ Differently from (\ref{redlocscheme}) this ring still has the nilpotent $A$ everywhere, so it does not have smooth points.

\begin{remark} Following what some authors in physics do, one could hand-impose an extra constraint
\beq A=a\, \chi^\al\chi_\al\,.\label{hand}\eeq
 The coefficient $a$  can  depend on $F$ but not on $\chi^\al$, because then $A$  would be identically zero.

For $n=2$ this followed from the constraints (\ref{phi2constraint}) only assuming that $F$ was invertible. Then $a$ was determined to be
$a=1/{4F}$.
  The same trick would not work for $n\geq 3$, so (\ref{hand}) is an extra constraint, not coming from (\ref{generalconstraint}), which nevertheless allows to solve trivially (\ref{generalconstraint}). Moreover, $a$ is arbitrary, so there are indeed solutions to $\Phi^3=0$  that do not solve $\Phi^2=0$. For $n=3,4,5,\dots$ the sets of solutions that we obtain in this way are identical.

  The constraint (\ref{hand}), when putting the odd variables to zero, gives $A=0$, which leaves us with the affine algebra $C^\infty(\C)$, something similar to what happened in Example \ref{nonaffine-ex}. The superalgebra would be $C^\infty(\C)[\chi_1, \chi_2]$, that is, the algebra of the  affine superspace
  $A^{1|2}$.

  \hfill$\square$

\end{remark}

This type of constraints appear in Ref. \cite{gh}.

\begin{remark} In Appendix \ref{susy-ap} we wrote the infinitesimal supertranslation algebra which acts on the affine superspace $A^{4|2}$. Actually, the supertranslation generators (\ref{supertranslationgen}) act on $A^{4|4}$, the {\it complexified Minkowski superspace}. There is of course a real form of this superspace and of  the supertranslation algebra which are the usual in physics. On the {\it chiral affine superspace} $A^{4|2}$ only acts the superalgebra generated by the generators  $P_\mu $ and $Q_\alpha$.

 The sets of equations (\ref{phi2constraint}) and (\ref{generalconstraint}) are supersymmetric since $\Phi^2=0$ is a supersymmetric constraint. While in the $n=2$ case the solution obtained by inverting $F$ is a supersymmetric solution, in the $n\geq 3$ case the solution obtained by imposing (\ref{hand}) is not supersymmetric. This can be checked by explicit calculation. Nevertheless, being the constraints supersymmetric, the space of solutions of  (\ref{generalconstraint}) has an action of the supertranslation algebra. It is then mandatory to keep the nilpotent $A$ with $A^{n-1}=0$ in order to preserve supersymmetry. Although we do not know yet the physical interpretation of such fields, it is remarkable that one is lead to maintain genuinely even nilpotents (that is, nilpotents that survive when putting the odd variables to zero)  in order to preserve the supersymmetry.

 \hfill$\square$

\end{remark}

In this situation we do not have enough odd variables for the problem with $n>3$ to be interesting. One can add more odd variables by going to extended supersymmetry. Physically, though,  it is more difficult to give meaning to superspace and superfields in extended supersymmetry. We could also consider real superfields, which have four real odd variables. Finally, one can use several superfields.

In the next section we see how we can satisfy cubic constraints with two superfields.

 \subsection{Cubic constraint with two superfields}\label{phi312-sec}

 Let us start with two superfields
 $$
\Phi_1=A_1 +\theta^\al \chi_{1\al}+\theta^\al\theta_\al F_1,\qquad
\Phi_2=A_2 +\theta^\al \chi_{2\al}+\theta^\al\theta_\al F_2\,.
$$
The quantities $A_1, A_2, F_1, F_2$ (even) and  $\chi_{1\al}, \chi_{2\al}$, $\al=1,2$ (odd) are coordinates in the superspace $A^{4|4}$. On these coordinates we want to impose the constraint $\Phi_1\cdot \Phi_2=0$:
$$\Phi_1\cdot \Phi_2=A_1A_2+\theta^\al(A_1\chi_{2\al}+A_2\chi_1^{\al})+
\theta^\al\theta_\al(A_1F_2+A_2F_1-\frac 12\chi_1^{\al}\chi_{2\al})=0\,,$$ which implies
\begin{align}
&A_1A_2=0,\nonumber\\
&A_1\chi_{2\,\al}+A_2\chi_{1\,\al}=0,\nonumber\\
&A_1F_2+A_2F_1-\frac 12(\chi_1\chi_{2})=0\,.\label{twofields}
\end{align}
To ease the notation we have writen $(\chi_1\chi_{2}):=\chi_1^{\al}\chi_{2\,\al}$.

This defines an affine superscheme $\cL$ with superalgebra
$$\cO(\cL):=C^\infty(\C^4)[\chi_{1\,\alpha},  \chi_{2\,\alpha}]/( A_1A_2,\, A_1\chi_{2\,\al}+A_2\chi_{1\,\al},\,A_1F_2+A_2F_1-\frac 12(\chi_1\chi_{2}))\,.$$
Let us compute the scheme, $\tilde \cL$. Setting to zero the  odd coordinates in (\ref{twofields}) we get
\beq
A_1A_2=0,\qquad A_1F_2+A_2F_1=0\,,\label{reducedconstraints2}
\eeq so
$$\cO(\tilde \cL)=C^\infty(\C^4)\big/\bigl( A_1A_2,\,A_1F_2+A_2F_1\bigr)\,.$$
We can now localize at $F_1\neq 0$ and  solve for $A_2$. Then (\ref{reducedconstraints2}) becomes
$$ F_2 A_1^2=0,\qquad
A_2=-F_1^{-1}{F_2} A_1\,. $$ If we are willing to restrict also to the points $F_2\neq 0$ we would  get
$$A_1^2=0,\qquad A_2=-F_1^{-1}{F_2} A_1\,,$$
so $A_1$ and $A_2$ are even nilpotents. The ring would become
$$\cO(\tilde \cL)_{F_1,\,F_2\neq 0}=C^\infty(\C^\times\times \C^\times\times \C)\big/\bigl( A_1^2\bigr)\,.$$ The scheme is not regular.

\bigskip

We now reintroduce the odd variables, but we localize only at $F_1\neq 0$. Then we can solve for $A_2$
$$A_2=-\frac{F_2}{F_1} A_1+\frac 1{2 F_1}(\chi_1\chi_2)\,.$$
Inserting into the first and  second equations in (\ref{twofields}) we get
\begin{align}
 & F_2 A_1^2-\frac 12(\chi_1\chi_2)A_1=0\nonumber\\
 &A_1\left(\chi_{2\,\al}-\frac {F_2}{F_1}\chi_{1\,\al}+\frac {1}{2F_1}(\chi_1\chi_2)\chi_{1\,\al}\right)=0\,.\label{a1}\end{align}
 If $F_2\neq 0$, one can see that $A_1^4=0$, so $A_1$ is nilpotent.
 \bigskip

One may  consider, as in (\ref{hand}), the extra condition  that $A_1$ is an even function of the odd variables,
\beq A_1=a(\chi_1^2)+ b(\chi_2^2)+c(\chi_1\chi_2)\,,\label{ansatz}\eeq with $a$, $b$ and $c$ coefficients that can also be functions of the other fields. As before, we stress that the meaning is that when putting the odd variables to zero we take $A_1=0$, which solves trivially $A_1^2=0$.

Inserting now the ansatz in (\ref{a1}), and after some calculations, we get
$$a=\frac{1}{4F_1}-\frac {F_2^2}{F_1^2}\,b,\qquad\qquad c=-\frac{2F_2}{F_1}\,b\,,$$ and $b$ is free.
 We have made use of the identities
$$(\chi_1\chi_2)\chi_{1\al}=-\frac 12 (\chi_1^2)\chi_{2\al},\qquad \qquad
 (\chi_1\chi_2)\chi_{2\al}=-\frac 12 (\chi_2^2)\chi_{1\al}\,.$$
 So the superfields become
 \begin{align}
 A_1&=\left(\frac{1}{4F_1}+\frac{F_2^2}{F_1^2}b\right)(\chi_1^2)+ b(\chi_2^2)-2\frac{F_2}{F_1}b
 (\chi_1\chi_2),\nonumber
 \\
 A_2&= \left(-\frac{F_2}{4F^2_1}
 -\frac{F_2^3}{F_1^3}b
 \right)(\chi_1^2)-
\left(2\frac{F_2^2}{F_1^2}b-\frac{F_2}{F_1}b(\chi_2^2)+\frac{1}{2F_1}\right)
(\chi_1\chi_2)\,,\label{a1a2}
 \end{align}
 with $\chi_1$ and $\chi_2$ free, $F_2$ free, and $F_1\neq 0$. Since the equations (\ref{twofields}) are symmetric under the exchange $1\rightleftarrows2$, one can obtain a similar solution by inverting $F_2$.

 If   $b=0$, the terms proportional to $b$ in  $\Phi_1$ vanish and  we get $\Phi_1^2=0$. With this choice we obtain for  $\Phi_2$
$$\Phi_2^2=-\frac 1{8F_1}\chi_1^2\chi_2^2\neq0\,.$$

It is not difficult to see that the  particular set of solutions with $b=0$ is equivalent to the system
\beq \Phi_1^2=0,\qquad \Phi_1\Phi_2=0\,,\label{phi1phi2}\eeq once we have inverted $F_1$. One obtains
$$A_1=\frac 1{4F_1}(\chi_1)^2,\qquad A_2= \frac {(\chi_1\chi_2)}{2F_1}-\frac{F_2(\chi_1)^2}{4F_1^2}\,,$$ and the remaining coordinates free (note that $F_2$ is not required to be invertible). This system has supersymmetry, since it has been cast in the supersymmetric form (\ref{phi1phi2}).

 \bigskip

 We  check now the cubic constraints for generic $b$:
 $$\Phi_1^2\Phi_2=\Phi_1\Phi_2^2=0$$ by virtue of $\Phi_1\Phi_2=0$. On the other hand, one gets also
 $$\Phi_1^3=0,\qquad\qquad  \Phi_2^3=0\,.$$
From the three equations in (\ref{generalconstraint}) with $n=3$, the first two ones are trivially satisfied because they have order greater than four in the odd variables and we only have four of them. The third one is only of order four in the odd variables and the terms must cancel exactly. It is not difficult to check that this happens for both superfields.

Nevertheless, the system $$\Phi_1\Phi_2=0,\qquad \Phi_1^3=0,\qquad \Phi_2^3=0$$ gives rise to a superscheme that is not regular. The constraints are
\begin{align*}
&A_1A_2=0, \qquad A_1\psi_2+A_2\psi_1=0,\qquad A_1F_2+A_2F_1-\frac 12(\psi_1\psi_2)=0,\\
&A_1^3=0,\qquad A_2^3=0,\qquad A_1^2\psi_{1\alpha}=0,\qquad A_2^2\psi_{2\alpha}=0, \\
&A_1\left(A_1F_1-\frac 12(\psi_1)^2\right)=0,\qquad A_2\left(A_2F_2-\frac 12(\psi_2)^2\right)=0\,.\end{align*}

The  scheme $\tilde\cL$ is given by the constraints
$$A_1A_2=0,\quad A_1F_2+A_2F_1=0,\quad A_1^3=0,\quad A_2^3=0,\quad F_1A_1^2=0, \quad F_2A_2^2=0\,.$$ Even restricting to $F_1\neq 0$, one obtains
$$A_2=-\frac {F_2}{F_1} A_1,\qquad A_1^2=0\,,$$ so a nilpotent remains that cannot be put directly to zero. The solutions obtained in (\ref{a1a2}) by imposing the ansatz  (\ref{ansatz}) do not reflect the whole solution space, and consequently they are not supersymemtric.

These constraints appeared in Refs. \cite{bfz, ks, fd}.

\subsection{Cubic constraint with an arbitrary number of superfields}\label{newconstraint-sec}

The system (\ref{phi1phi2}) can be generalized by adding more chiral superfields.  We consider the following system :
\begin{align*}
&X=A+(\theta\chi)+\theta^2 F, &&X^2=0 \\
&Y_i=A_i+(\theta\psi_{i})+\theta^2 F_i, &&XY_i=0
\end{align*}
for $i=1, \dots ,n$ and $n$ arbitrary.
From (\ref{phi2constraint}) and  (\ref{twofields}) the constraints are equivalent to the system
\begin{align*}
 &A^2=0,&& A\chi_\alpha=0, &&4AF-\psi^2=0\\
&AA_i=0,
&&A\psi_{i\,\al}+A_i\chi_{\al}=0,
&&AF_i+A_iF-\frac 12(\chi\psi_{i})=0\,.
\end{align*}
Putting the odd variables to zero, the constraints  become
\begin{align*}
 &A^2=0,&&AF=0\\
&AA_i=0,
&&AF_i+A_iF=0\,.
\end{align*}
Localizing at $F$ invertible we can solve
$$A=0, \qquad A_i=0\,,$$ which means that the  scheme has a smooth part
$$\C[\tilde\cL]_{F\neq 0}=\C^\infty(\C^\times\times \C^n)\,.$$
Reinserting the fermions we get
$$A=\frac{\psi^2}{4F},\qquad A_i=-\frac{\chi^2 F_i}{4F^2}+\frac 1{2F}(\chi\psi_{i})\,,$$ and the remaining equations are satisfied trivially.

In this case we can use the same method than in Section \ref{phi2-sec}. We perform a change of variables
\begin{align*}A'=4AF-\chi^2,\qquad A'_i=AF_i+A-iF-\frac 12\chi\psi_i,\\
F'=F,\qquad F'_i=F_i,\qquad \chi'=\chi,\qquad \psi'_i=\psi_i\,,\end{align*}
 whose Jacobian is invertible if $F$ is so. The constraints are then
 $$A'=0, \qquad A'_i=0\,,$$ and  the superscheme  at $F$ invertible becomes
 $$\C[\cL]_{F\neq 0}\simeq\C^\infty(\C^\times \times \C^n)[\chi_\al, \psi_{i\,\al}]\,,$$ with $\al=1,2,$ and $i=1,\dots , n$.
The superfields $Y_i$, $i=1, \dots, n$ satisfy
\beq Y_iY_jY_k=0,\qquad \forall \;i,j,k=1,\dots ,n\,.\label{3spf}\eeq
In order to prove this we have used the following Fierz identity:
$$(\psi_1\psi_2)\psi_3+(\psi_3\psi_1)\psi_2+(\psi_2\psi_3)\psi_1=0\,.$$
 We note that this is not the most general solution to (\ref{3spf}). For $n=1$ the solution presented in (\ref{a1a2}) with $\Phi_1=X$ and $\Phi_2=Y$  is  a more general one ($b\neq 0$).

 The case $n=3$ was presented in Refs. \cite{vw, kvw}.

 \section{A non algebraic constraint}\label{nonalgebraic-sec}

 In this section we are going to consider both, chiral and antichiral superfields. Up to now we were considering only chiral superfields, so the description of Section \ref{algebraic-sec} was the simplest one. Moreover, the spacetime coordinates would not appear explicitly in the discussion of the constraints, so it was as if spacetime was reduced to a point.
 In this section we will not be in that case anymore and the spacetime variables would play a role.

 We shall start with a (complexified) super spacetime $M=A^{4|4}$, with  $$x^\mu,\quad \mu=0,\dots, 3,\qquad \theta^\alpha, \; {\bar\theta}^{\dot\alpha},\quad \alpha,\dot\alpha=1,2$$ being its global coordinates. As the notation suggests, $\theta^\alpha$ and $\bar{\theta}^{\dot\alpha}$ are related by an antilinear involution that defines  the standard  (real) super Minkowski space, for which the coordinates $x^\mu$ are real. But for our purposes it is better to keep the super spacetime complexified, so $\theta^\alpha$ and ${\bar\theta}^{\dot\alpha}$ are independent odd coordinates and $x^\mu $ are even, complex coordinates.

According to what we expressed at the end of Section \ref{superfields-sec}, we will now work with families of morphisms  $S\times M\rightarrow \C$ parametrized by $S$. In terms of the algebras we will have, as before
$$\begin{CD}\cO(\C)@>>>\cO(S)\otimes\cO(M)\\
z@>>> \Phi\,,
\end{CD}$$
where
\begin{align}\Phi=\,&F+(\theta\chi)+(\bar\theta\bar\chi)+(\theta\theta)M+(\bar \theta\bar \theta)N+ (\theta\sigma^\mu\bar\theta)V_\mu\nonumber \\ &+(\theta\theta)(\bar\theta\bar \lambda)+(\bar\theta\bar\theta)(\theta\Psi) +(\theta\theta)(\bar\theta\bar\theta) D\,.\label{superfieldbig}\end{align}
If there were no dependence on the coordinates $x^\mu$ of spacetime, then the superspace of superfields would be $\cH=A^{8|8}$, with coordinates
$$(F, M, N, V_\mu, D \;|\; \phi^\alpha,\bar\chi^{\dot\alpha},\bar \lambda^{\dot\alpha},\psi^\alpha),\qquad \mu=0,\dots 4,\quad \alpha,\dot\alpha=1,2\,.$$
The symbols $\sigma^\mu$ stand for the Pauli matrices and the notation for the spinor (odd) fields is explained in Appendix \ref{spinors-ap}.

When we consider that spacetime is not a point, the coordinates above depend on $x^\mu$ and are the {\it component fields} of the superfield. This suffices to describe the family of maps parametrized by $S$.

We consider now two odd derivations on $\cO(M)=\cO^{4|4}(\C^4)$
$$ D_\alpha=\partial_\alpha+\ri \sigma^\mu_{\alpha\dot\alpha}{\bar \theta}^{\dot\alpha}\partial_\mu,\qquad
\bar D_{\dot\alpha}=-\partial_{\dot\alpha}-\ri \theta^\alpha\sigma^\mu_{\alpha\dot\alpha}\partial_\mu\,,
$$
where
$$\partial_\alpha=\frac{\partial}{\partial\theta^\alpha}, \qquad
\partial_{\dot\alpha}=\frac{\partial}{\partial{\bar\theta}^{\dot\alpha}},\qquad \partial_\mu =\frac{\partial}{\partial x^\mu}\,.$$
$D_\alpha$ and $\bar D_{\dot\alpha}$ are indeed the right invariant vector fields of the action of the supertranslation generators (see Appendix \ref{susy-ap}).

The derivations above can be applied to $\cO(S)\otimes\cO(M)$, provided that the rule of signs is appropriately taken into account.

We now define chiral superfields in terms of a constraint on superfields, as it is done in physics. We say that the superfield $\Phi$ (\ref{superfieldbig}) is chiral if $$\bar D_{\dot\alpha}\Phi=0\,.$$
From the fact that
$$\bar D_{\dot\alpha}(x^\mu+\ri  (\theta\sigma^\mu{\bar \theta})=0\,,$$
it is  easy to see that, under the change of variables
$$(x^\mu,\, \theta^\alpha,\, \bar{\theta}^{\dot\alpha})\;\longrightarrow \;
(y^\mu= x^\mu+\ri  (\theta\sigma^\mu{\bar \theta}), \,\theta^\alpha,\, \bar{\theta}^{\dot\alpha}),\qquad (\theta\sigma^\mu{\bar \theta}):=\theta^\alpha\sigma^\mu_{\alpha\dot\alpha}{\bar \theta}^{\dot\alpha}\,,$$
a chiral superfield $\Phi^\mathrm{ch}$ can be written as
$$\Phi^\mathrm{ch}=A^\mathrm{ch}(y)+\left(\theta\chi^\mathrm{ch}(y)\right)+
\theta^2 F^\mathrm{ch}(y)\,.$$
Instead, under the change of variables
$$(x^\mu,\, \theta^\alpha,\, \bar{\theta}^{\dot\alpha})\;\longrightarrow \;
(\bar y^\mu= x^\mu-\ri  (\theta\sigma^\mu{\bar \theta}), \,\theta^\alpha,\, \bar{\theta}^{\dot\alpha})\,,$$ the antichiral superfield
$$D_\alpha X^\mathrm{ach}=0\,,$$ is expressed as
$$X^\mathrm{ach}=A^\mathrm{ach}(\bar y)+\left (\bar\theta\chi^\mathrm{ach}(\bar y)\right)+
\bar\theta^2F^\mathrm{ach}(\bar y)\,.$$
$y^\mu$ and $\bar y^\mu$ are related by complex conjugation (as the notation suggests).
The complex conjugate of a chiral superfield is an antichiral superfield and we also have
\beq \bar y^\mu=y^\mu -2\ri  (\theta\sigma^\mu{\bar \theta})\,.\label{cachange}\eeq

\bigskip
Let us now consider now two chiral superfields
$$X=A(y)+\theta\chi(y)+\theta^2 F(y),\qquad
Y=B(y)+\theta\psi(y)+\theta^2 G(y)\,,$$  and assume that $X^2=0$ as in Section \ref{phi2-sec}. Then, if $F$ is invertible
$$X=\frac {\chi(y)^2}{4F(y)}+\theta\chi(y)+\theta^2 F(y)\,.$$
We write the complex conjugate of $X$ as
$$\bar X=\frac {\bar\chi(\bar y)^2}{4\bar F(\bar y)}+\bar\theta\bar \chi(\bar y)+\theta^2 \bar F(\bar y)\,,$$ where the notation (the usual one in physics) means
$$\bar A(\bar y):= \overline{A(y)},\qquad \bar \chi_{\dot\alpha}(\bar y):= \overline{\chi_\alpha(y)},\qquad \bar F(\bar y):= \overline{F(y)}\,.$$

The constraint that we intend to impose is \cite{ks}
\beq\bar X Y= \hbox{antichiral}\,.\label{acconstraint}\eeq In order to do that, one writes the superfield $Y$ in terms of the variable $\bar y^\mu$, using (\ref{cachange}) and expanding in Taylor series, which is finite because of the nilpotency of the odd variables. One gets
$$Y=B+ \theta \psi+\theta^2 G+2\ri\partial_\mu B (\theta\sigma^\mu\bar \theta)-\ri\theta^2(\partial_\mu\psi\sigma^\mu\bar \theta)\,,$$ where all the component fields are evaluated at $\bar y^\mu$. In order to impose the constraint (\ref{acconstraint}), the only components of $\bar XY$ that can survive are the ones proportional to $1$, $\bar\theta^\alpha$ and $\bar\theta^2$. After some calculations we get (recall that $F$ is invertible)

\begin{align}
&\bar \chi^2\psi_\alpha=0, &&\bar\chi^2G=0\nonumber\\
&\frac {\partial^2 B}{4\bar F}\bar\chi^2+\frac {\ri}{2}(\partial_\mu\psi\sigma^\mu\bar \chi)+ \bar F G=0, &&
-\ri\partial_\mu B\sigma^\mu_{\alpha\dot\alpha}\bar \chi^{\dot\alpha} +\psi_\alpha \bar F=0,\nonumber \\
&-\ri(\partial_\mu\psi^\alpha\sigma^\mu_{\alpha\dot\alpha})\frac{\bar \chi^2}{4\bar F}+ \bar\chi_{\dot\alpha}G=0,&& \frac {\ri\partial_\mu B\chi^2}{2\bar F}\sigma^\mu_{\alpha\dot\alpha}+\bar\chi_{\dot\alpha}\psi_\alpha=0\,.\label{diffconstraints}
\end{align}
Putting the odd variables to zero, we get  $G=0$ and $B$  undetermined. The constraints are algebraic and the scheme we obtain  is
$$\cO(\tilde \cL)\simeq C^\infty (\C^\times\times \C)\,.$$
The full constraints can be considerably simplified using the fact that $F$ is invertible. For example, one can isolate $G$ and $\psi_\alpha$
$$G=-\frac {\partial^2 B}{4\bar F^2}\bar\chi^2-\frac {\ri}{2\bar F}(\partial_\mu\chi\sigma^\mu\bar \chi),\qquad \psi_\alpha=
\ri\partial_\mu B\sigma^\mu_{\alpha\dot\alpha}\bar \chi^{\dot\alpha}\,,$$ and the remaining constraints are satisfied. Although the constraints involve derivatives, the superscheme can be given algebraically, in its complex version, as
$$\cO(\cL)\simeq C^\infty (\C^\times\times \C)[\chi_\alpha,\bar \chi_{\dot\alpha}]\,.$$

This example is also illustrative of the properties of the odd fields. All through the calculations one has to assume that an odd field, say $\chi_\alpha$, and its spacetime derivative $\partial_\mu\chi_\alpha$ have a product that is different from zero. Finally, the superscheme has a smooth part, on which the supersymmetry transformations have a well defined action.

 \section{Observables and nilpotent variables}\label{observables-sec}
 Let $F$ be an algebra  and consider the scheme $\uspec(F)$. Let $\fp$ be a prime ideal in $F$, so $\fp\in |X|=\rspec(F)$ and consider the quotient $F/\fp$.  This is an integral domain (the product of two non zero elements is a non zero element). Moreover, if we consider the localization of $F$ over $\fp$,
 $$F_\fp=\left\{ \frac{f}{g}\; |\; f, g\in F,\;g(\fp)\neq 0  \right\}\,,$$
 and we quotient with the ideal $F_\fp \cdot \fp$, the result $\kappa(\fp):=F_\fp/(F_\fp \cdot \fp)$ is a field, since every non zero element in $\kappa(\fp)$ has an inverse. The field $\kappa(\fp)$ is called the {\it residue field of $|X|$ at $\fp$}.

 \begin{example} $\hfill$

 \begin{enumerate}

 \item We consider the ring of polynomials in one variable $F=\C[x]$. The prime ideals of $F$  are of the form $\fp_a=( x-a)$, $a\in \C$ or the ideal $( 0)$. It is not difficult to see that the residue field at $\fp_a$ is $\kappa(\fp_a)\cong \C$ and  $\kappa(( 0))$ is the field of rational functions.

 \item If the ground field is $\R$, we have that every irreducible polynomial in $\R[x]$ generates a prime ideal. We still have the maximal ideals as $\fp_a=( x-a)$ that give all the points of $\R$. At them, the residue field is $\kappa(\fp_a)\cong\R$. But, for example, the irreducible polynomial $x^2+1$  also generates a prime ideal. It is not difficult to realize that the elements of the residue field are of the form $a+ xb$, with  $a, b\in \R$ and  $x$ such that $x^2+1=0$, so $\kappa(( x^2+1))\cong\C$.

 \end{enumerate}

 \hfill$\square$
 \end{example}

As we have seen, affine schemes have residue fields that can  vary from point to point. Let $F$ be an algebra, not necessarily affine. For every element $f\in F$ we can define a `function' on $\rspec(F)$ with values in the residue field via the canonical maps
$$\begin{CD}
F@>>>F_{\fp} @>>>\kappa(\fp)\\f@>>>f@>>>f(\fp)\,.
\end{CD}$$
In an affine variety, with affine algebra $F$, one recovers in this way the original interpretation of $F$ as the algebra of functions on the algebraic variety. The same holds in the case of differentiable manifolds and smooth functions.

 If $F$ contains a nilpotent element, say $n$, then $n\in \fp$ for all prime ideals $\fp$ so $n(\fp)=0$. In other words, $n$ is sent to the zero function and one cannot reproduce the original algebra $F$  starting from an algebra of   functions on $\rspec(F)$. This is something that we already knew  (see Example \ref{nonaffine-ex}), but now we can read it from a physical point of view.

 \smallskip

 A classical mechanics system is commonly described in terms of a symplectic manifold called {\it phase space}, whose points represent the  possible states of the system.  Classical observables are smooth functions on phase space. There is a special observable, the Hamiltonian, which governs the time evolution  of the system: given the initial state in an instant of time $t_0$, the system evolves in future times $t$ by following the integral curve of the hamiltonian vector field associated to the Hamiltonian, passing through the initial state.

 This picture can be more or less carried over classical field theory by substituting the phase space for an infinite dimensional  space of maps from spacetime to a target manifold or of sections of some bundle over spacetime, which are the {\it fields}. Most of the time one uses variational calculus to approach classical field theory instead of trying to give some comprehensive study of these infinite dimensional spaces, which can be very involved. Nevertheless, the idea of observable is mimicked from classical mechanics: observables are (a special class of) functionals on the space of fields.  The time evolution of the system is also governed by some partial differential equations (usually up to second order) for the fields.

 The idea that we want to convey can be already understood at the classical mechanics level. Suppose that we want to generalize the classical phase space to some sort of affine scheme whose algebra contains nilpotents. The usual way to obtain `numbers' (results of a measurement) from the sections of the scheme is by the evaluation procedure explained above. Nilpotent elements go to zero by this map, so they do not represent observables.

 Now, one could do the same for a superscheme: even a smooth supermanifold or regular algebraic variety contains nilpotents generated by the odd elements which could not be seen in any `classical' measurement. What we are affirming is that in this interpretation of the observables, classical, odd degrees of freedom could not be seen in experiments.

 \smallskip

 In quantum mechanics things are very different. States are rays in a Hilbert space and observables are hermitian operators on it. The results of measurements are eigenvalues of these operators, and they appear with a probability distribution determined by the Hilbert space state. The algebra of operators on a Hilbert space is non commutative, so sometimes it is said, very  roughly, that  `quantizing' a system corresponds to substitute the commutative algebra of observables by a non commutative one such that when taking the limit $\hbar\rightarrow 0$ the original commutative algebra is recovered.

   Let us consider the simplest case possible, a two dimensional phase space $\R^2$ with canonical coordinates $(q,p)\in \R^2$ and symplectic form $\rd q\wedge\rd p$. The induced Poisson bracket on the coordinates is
 \beq\{q, p\}_-=1\,.\label{canonicalcr}\eeq
 As a quantum system, one considers the Hilbert space of square integrable functions on the variable $q$, $\cL^2(\R)$. One considers the  position and momentum operators:
 $$Qf(q)=q f(q),\qquad Pf(q)=-\ri\hbar\frac{\partial f}{\partial q} \,,$$ whose commutation rule is
 $$[Q, P]_-=\ri\hbar \,\rid\,.$$
 Taking  $\hbar\rightarrow 0$ the commutation relation is reverted to the commutativity of $q$ and $p$ as ordinary functions on phase space. The fact that the term of  order one in $\hbar$ is proportional to the Poisson bracket (substituting the constant function `1' by the identity) is not casual, but a requirement of the quantization.

 Let us assume  that the phase space is now substituted by a superspace, for example $A^{2|2}$, with superalgebra $C^\infty(\R^2)\otimes\wedge[\theta, \pi]$. There is also a super Poisson structure on it:
 $$\{\theta, \pi\}_+= 1,\qquad  \{q, p\}_-=1\,, $$ and the rest zero.  The Poisson bracket of two odd quantities is symmetric.

 As in the non super case, the superalgebra $C^\infty(\C^2)\otimes\wedge[\theta, \pi]$ admits a deformation with parameter $\hbar$.

 Let us focuss now on the deformation of the Grassmann algebra $\wedge[\theta,\pi]$. Mimicking the procedure with the even variables, we get a  non commutative superalgebra with generators $\Theta$ and $\Pi$ satisfying the commutation rules
 $$[\Theta,\Pi]_+=\ri \hbar\, \rid \qquad[\Theta, \Theta]_+=0,\qquad [\Pi,\Pi]_+=0\,, $$ which is the algebra of the odd quantum oscillator with its creation and destruction operators. From them, one can construct the Hilbert space and the quantum observables.

 There is a linear change of variables
 $$ \Gamma=\Theta+\ri \Pi,\qquad \Upsilon=\Theta-\ri \Pi\,,$$ which shows that this algebra is isomorphic to the Clifford algebra $C(1,1)$ \cite{bm,llf,ll}:
 \beq[\Gamma,\Gamma]_+=\hbar \, \rid ,\qquad [\Upsilon,\Upsilon]_+=-\hbar \, \rid ,\qquad [\Gamma,\Upsilon]_+=0\,.\eeq  The key point here is the symmetry of the super Poisson bracket.

 \smallskip

 The conclusion is that, in their quantum version, odd variables can give rise to meaningful observables. The classical limit $\hbar\rightarrow 0$ leaves the superspace mathematical structure, but is does  not produce classical, odd observables.

\section*{Acknowledgements}

I would like to thank the Department of Physics and Astronomy and the Department of Mathematics of UCLA for their kind hospitality during the realization of this work.

I'm indebted to S. Ferrara and V. S. Varadarajan for invaluable discussions. I also would like to mention an inspirational conversation about fermions and observables that I maintained many years ago at CERN with R. Stora. I specially want to thank P. Deligne for suggesting an approximation to superfields more natural and elegant than our the previous version, that considered only Grassmann algebras.

This work has been supported in part by grants  FIS2011-29813-C02-02, FIS2014-57387-C3-1, FIS2017-84440-C2-1-P and  SEV-2014-0398 of the Ministerio de Econom\'{\i}a y Competitividad (Spain) and European Funds for Regional Development, European Union -A way to construct Europe- and  by Generalitat Valenciana through the project SEJI/2017/042.

\appendix

\section{Some basic definitions}\label{basic-ap}

\begin{definition}\label{extensions-def} {\sl Extension of scalars.} Let $\cA$ and $\cA'$ be two commutative superalgebras over $k$  and let $f:\cA\rightarrow \cA'$ be a morphism of superalgebras. Then $\cA'$ is an $\cA$-module with action
 \begin{align*}\begin{CD}\cA\otimes \cA'@>>>\cA'\\a\otimes a'@>>> f(a)a'\,.\end{CD}\end{align*} If $M_\cA$ is an $\cA$-module  we can define an $\cA'$-module as the tensor product $$M_{\cA'}:=\cA'\otimes_\cA M_\cA=\cA'\otimes M_\cA\big/\bigl(\, a'\otimes a\cdot m-a'f(a)\otimes m\,\bigr)$$
 with  action
\begin{align*}\begin{CD}\cA'\otimes M_{\cA'}@>>>M_{\cA'}\\
  b'\otimes [a'\otimes m]@>>>[(b'a')\otimes m]\,.\end{CD}\end{align*}
We say that $M_{\cA'}$ is the extension of scalars of $M_\cA$ to $\cA'$.

\hfill$\square$
\end{definition}

\begin{definition}\label{sheaf-def}{\sl Sheaf over a topological space.}
Let $|X|$ be a topological space. A {\it presheaf} on $|X|$ assigns to each open set $U\subset |X|$ a set $\cF(U)$ (it can be an abelian group, an algebra, a superalgebra, a module, ...) and to every pair of open sets $U\subset V\subset |X|$ a {\it restriction map}
$$\begin{CD}\cF(V)@>\res_{V, U}>>\cF(U)\end{CD}$$ satisfying
\begin{enumerate}
\item $\res_{U, U}=\rid$.
\item $\res_{V, U}\circ \res_{W, V}= \res_{W, U}$ for al  $U\subset V\subset W\subset |X|$.

\end{enumerate}

The elements of $\cF(U)$ are called {\it local sections} of $\cF$ over $U$ and the elements of $\cF(|X|)$ are called {\it global sections}.

A presheaf is a {\it sheaf} if it satisfies the condition that, for each open covering $\{U_\alpha\}_{\alpha\in A}$ of an open set $U$ (in particular, of the total space $|X|$),  and each collection of elements $\{f_\alpha\in \cF(U_\alpha)\}_{\alpha\in A}$ such that
$$\res_{ U_\alpha, U_\alpha\cap U_\beta} (f_\alpha )=
\res_{ U_\beta, U_\alpha\cap U_\beta} (f_\beta ),\qquad \forall \alpha,\beta \in A\,,$$
there exists a unique element $f\in \cF(U)$ such that
$$\res_{U, U_\alpha}(f)=f_\alpha,\qquad \forall \alpha\in A\,.$$

\hfill$\square$

\end{definition}

\begin{example}{$\hfill$}

\begin{enumerate}

\item Continuous, differentiable, real analytic or complex analytic functions on a topological space are all sheaves of algebras.

\item  Sections of a vector bundle over a topological space are a sheaf of modules over some algebra of functions.

\item Constant functions over a topological space are, generically, only a presheaf. If the space is connected, then the sheaf condition is satisfied. Also, on a not necessarily connected space, one can define the sheaf of {\it locally constant functions}, that is, functions that are constant on an open neighborhood of each point.

\end{enumerate}
\hfill$\square$
\end{example}

\begin{definition} \label{stalk-def} {\sl Stalk of a sheaf  over a point.} Let $\cF$ be a sheaf of abelian groups (all the sheaves that we use are so) over the topological space $|X|$. Let $x\in |X|$. The {\it stalk of $\cF$ at $x$}, denoted as $\cF_x$ is the direct limit (see for example Ref.  \cite{eh})  of the family of abelian groups $\cF(U)$ running over all neighborhoods $U$ of $x\in |X|$.

\hfill$\square$

\end{definition}

\begin{definition} \label{sheafmor-def} A {\it morphism of sheaves} $\varphi: \cF\ra\cG$ over the same topological space $|X|$ is a collection of maps $\varphi_U: \cF(U)\ra\cG(U)$, $U\subset_\mathrm{open} |X|$ such that for every pair of open sets $V\subset U\subset |X|$ the diagram
$$\begin{CD}
\cF(V)@>\varphi_V>>\cG(V) \\
@V\res_{V,U}VV @VV\res_{V,U}V\\
\cF(U)@>\varphi_U>>\cG(U)
\end{CD}$$
commutes

\hfill$\square$

\end{definition}

\section{Notation for spinors}\label{spinors-ap}

 Let $\theta^\al$ with  $\al=1,2$ odd coordinates in some affine superspace. As customary, we define
$$ (\ep_{\al\be}):=\begin{pmatrix}0&1\\-1&0\end{pmatrix},\qquad
(\ep^{\al\be}):=\begin{pmatrix}0&-1\\+1&0\end{pmatrix},\qquad \ep^{\al\be}\ep_{\be\ga}=\delta^\al_\ga\,,$$ where sum over repeated indices is understood. We also define
$$
\theta_\al:=\ep_{\al\be}\theta^\be,\qquad\hbox{so}\qquad  \theta^\al=\ep^{\al\be}\theta_\be\,.$$ In general, for any two pairs of odd quantities, $\theta^\al$ and $\psi^\al$, anticommuting among them
 $$\theta^\al\theta^\be=-\theta^\be\theta^\al,\qquad
\psi^\al\psi^\be=-\psi^\be\psi^\al,\quad\theta^\al\psi^\be=-\psi^\be
\theta^\al\,,$$ one has
$$\theta^\al\psi_\al=
\psi^\al\theta_\al,\qquad \theta^\al\theta_\al=2\theta^1\theta^2\,.$$
One also denotes $(\theta\psi)=\theta^\al\psi_\al$.

 One can prove the so called {\it Fierz identities} \begin{align}&\theta^\al\theta^\be=-\ep^{\al\be}\theta^1\theta^2=-\frac 12\ep^{\al\be}\theta^\ga\theta_\ga,\nonumber\\
&\theta_\al\theta_\be=\ep_{\al\be}\theta_1\theta_2=
\frac 12\ep_{\al\be}\theta^\ga\theta_\ga,\nonumber\\&
(\psi\theta)\chi+(\chi\psi)\theta+(\theta\chi)\psi=0\,.\label{fierz}
\end{align}

\bigskip

The Pauli matrices are
$$
\sigma_0=\begin{pmatrix} 1 & 0 \\  0  &  1 \end{pmatrix}=\id, \quad
\sigma_1=\begin{pmatrix} 0 & 1 \\  1  &  0\end{pmatrix}, \quad
\sigma_2=\begin{pmatrix} 0 & -\ri \\ \ri   &  0 \end{pmatrix}, \quad
\sigma_3=\begin{pmatrix} 1 & 0 \\  0  &  -1 \end{pmatrix}\,,
$$ and they satisfy  the relations
$$\sigma_i\sigma_j=\delta_{ij}\id+\ri \epsilon_{ijk}\sigma_k\,,\qquad i,j,k=1,2,3\,,$$
where, as usual, $\ep_{ijk}$ is the totally antisymmetric tensor with $\ep_{123}=1$.

\section{Supersymmetry transformations.} \label{susy-ap} We give here the relations of the supertranslation algebra acting on the Minkowski superspace.     We follow the conventions of Ref. \cite{fz}. A basis of the  supertranslation Lie algebra is given by
\begin{align}
&Q_\al, \quad \bar Q_{\dal}, &&\al,\dal=1,2\qquad&&\hbox{(odd), }\nonumber\\
&P_\mu,\quad &&\mu=0,\dots, 3 \qquad  & & \hbox{(even)}\,.\label{supertranslationgen}
\end{align}
The standard real form makes $P_\mu$ real and $\bar {Q}_{\dal}$ the complex conjugate of $Q_\al$.
The commutation relations among the generators (\ref{supertranslationgen}) are
$$\{Q_\alpha,\bar Q_{\dbe}\}=2(\sigma_\mu)_{\al\dbe}P^\mu\,,$$ and the rest zero (the Pauli matrices are listed in Appendix \ref{spinors-ap}). The action of the supertranslation algebra on chiral superfields is as follows:
Let $\xi^\al$, $\bar\xi^\dal$ denote the odd supertranslation parameters and $a^\mu$ the even ones. The infinitesimal transformations on the component fields $A$, $\psi_\al$ and  $F$
\begin{align}
\delta_{a} (\,\cdot\,)&= a^\mu \ri\partial_\mu (\,\cdot\,) \qquad \hbox{(applied to $A$, $\psi_\al$ and  $F$)},\nonumber\\
\delta_\xi A&=\xi^\alpha\psi_\alpha,\nonumber\\
\delta_\xi\psi_\alpha&=2(\sigma^\mu)_{\al\dbe}\bar \xi^{\dbe}
\ri\partial_\mu A+2F\xi_\al,\nonumber\\
\delta_\xi F&=-\ri\partial_\mu\psi^\alpha (\sigma^\mu)_{\al\dbe}\xi^\dbe\,.\label{susytrans}
\end{align}
Acting with  $\bar {Q}_{\dal}$ on a chiral superfield gives a non chiral superfield, so only $P_\mu$ and $Q_\al$ have a well defined action on the set of chiral superfields. The full super translation algebra acts, for example, on the space of real superfields, which describe a multiplet of supersymmetry with maximal spin 1 (a vector field potential).

\end{document}